\documentclass[12pt]{article}

\usepackage{geometry}
\usepackage{setspace}
\usepackage{needspace}
\usepackage{todonotes}
\usepackage{soul,xcolor}

\usepackage{amssymb,amsmath,amsfonts,latexsym,amsthm,bbm}
\usepackage{mathptmx} 
\usepackage{graphicx}
\usepackage{color}
\usepackage[utf8]{inputenc}
\usepackage[english]{babel}
\usepackage{float}
\usepackage{tikz}
\usepackage{natbib}
\usepackage[colorlinks=true,linkcolor=black,anchorcolor=black,citecolor=blue,urlcolor=blue]{hyperref}
\usepackage{footnotebackref}
\usepackage{microtype}
\usepackage{csquotes}

\newtheorem{theorem}{Theorem}
\newtheorem{proposition}{Proposition}
\newtheorem{propositionprime}{Proposition}

\newtheorem{corollary}{Corollary}
\newtheorem{lemma}{Lemma}
\newtheorem{definition}{Definition}

\newtheorem{claim}{Claim}

\newtheorem{example}{Example}

\newcommand{\E}{\mathbb{E}}
\newcommand{\MPC}{\text{MPC}}
\newcommand{\DET}{\text{DET}}
\newcommand{\R}{\mathbb{R}}
\newcommand{\Prob}{\mathbb{P}}
\newcommand{\1}{\mathbf{1}}

\DeclareMathOperator{\supp}{supp}

\usepackage{caption}
\usepackage{subcaption}
\usepackage{booktabs}
\usepackage{multirow}
\usepackage{rotating}
\usepackage{tabularx}

\usepackage{pgfplots}
\pgfplotsset{compat=newest}
\usetikzlibrary{plotmarks}
\usetikzlibrary{arrows.meta,decorations.pathreplacing,patterns}
\usepackage{grffile}

\usepackage[title,titletoc,toc]{appendix}

\makeatletter
\newcommand\urlfootnote@[1]{\footnote{\url@{#1}}}
\DeclareRobustCommand{\urlfootnote}{\hyper@normalise\urlfootnote@}
\makeatother

\begin{document}
	
	\title{Deterministic Implementation of Mean-Preserving Contractions\thanks{I gratefully acknowledge financial support from the National Natural Science Foundation of China (Nos.\ 72403185 and 72433003). I used ChatGPT (GPT-6.1 Sol) to improve the clarity of writing and assist with checks of the draft's internal consistency. All remaining errors are mine.}}
	\author{Junze Sun\thanks{School of Economics and Management, Wuhan University. Email: \href{mailto:sunjunze@gmail.com}{sunjunze@gmail.com}.}}
	\date{This version: \today. Check \href{https://www.dropbox.com/scl/fi/2ob4wfhhj4h6o5uuj1dg0/Draft_MPC.pdf?rlkey=zk7jybfod108j3vw721x1t48g&st=6uwkz4n1&dl=0}{here} for the latest version. \\ \vspace*{1ex} Preliminary, comments are welcome.}
	\maketitle
	
	\begin{abstract}
		\noindent We study whether randomization is needed to generate a prescribed distribution of posterior means for a one-dimensional state. Drawing on backward Monge martingale transport, we show that every mean-preserving contraction of an atomless prior with finite first moment admits a deterministic implementation. Consequently, restricting attention to deterministic signals leaves the designer’s value unchanged whenever objectives and constraints depend only on the posterior-mean distribution. We decompose implementation into independent problems on non-separating intervals and describe a constructive procedure based on barcode transport. We also characterize when lower or upper contour sets can take a nested-interval form, and provide applications to monopolistic pricing with buyer-optimal learning and robust persuasion with ambiguous receiver preferences. 
		
	\end{abstract}
	\setcounter{page}{0}
	\thispagestyle{empty}
	\newpage

	\section{Introduction}
	Many economically important design problems can be formulated as
	\begin{equation}\label{Eq: General_Problem}
		\max_{G}{V(G)} \qquad \text{subject to} \qquad G\in\mathcal{A}\subseteq\MPC(F)
	\end{equation}
	where $\MPC(F)$ denotes the set of mean-preserving contractions of a distribution $F$ on the real line. In information design, this formulation has a natural interpretation: $F$ is the prior distribution of a one-dimensional state, and $G$ is the distribution of posterior means generated by a statistical experiment. Bayes plausibility is equivalent to the requirement that $G\in\MPC(F)$. In standard linear persuasion problems, the designer's objective is linear in $G$ and $\mathcal A=\MPC(F)$ \citep{kamenica2011bayesian,gentzkow2016rothschild,Dworczak2019,Kolotilin2022censorship}. The same formulation also accommodates nonlinear objectives and additional restrictions on feasible distributions. Examples include monopolistic pricing with buyer-optimal learning, where the seller's optimal price depends on $G$ \citep{Roesler2017buyer}, and robust persuasion, where the designer evaluates $G$ against the worst-case distribution of receiver preferences \citep{SapiroGheiler2024}.
	
	This reduced-form approach allows one to identify a desirable posterior-mean distribution without specifying the experiment that generates it. It leaves open, however, a concrete implementation question: given a prior $F$ and a target distribution $G\in\MPC(F)$, how can one construct a signal that implements $G$ from $F$? A general experiment allows randomization conditional on the realized state. Such randomization may be difficult to carry out or to make credible in practice. A deterministic experiment instead assigns each state a single message and can thus be described by a partition of the state space. We ask whether this restriction reduces the set of feasible posterior-mean distributions and, if it does not, how to construct an implementing partition.
	
	We show that, whenever the prior $F$ is atomless and has a finite first moment, every mean-preserving contraction of $F$ has a deterministic implementation. Thus, deterministic signals and unrestricted stochastic experiments generate exactly the same set of posterior-mean distributions. This feasible-set equivalence leaves the designer's value unchanged for the entire class of problems in \eqref{Eq: General_Problem}. It requires neither linearity of $V$ nor any additional restriction on $\mathcal A\subseteq\MPC(F)$, as long as the objective and constraints depend on the experiment only through the posterior-mean distribution. Whenever an optimal distribution exists under unrestricted experiments, that same distribution can be implemented deterministically.
	
	The feasible-set equality follows directly from the backward-Monge theorem of \cite{Nutz2024martingale}. The classical theorem of \cite{Strassen1965existence} guarantees that every mean-preserving contraction can be generated by some experiment. The refinement in \citet[Theorem~2.1]{Nutz2024martingale} ensures that, when the prior is atomless, the induced posterior mean can be chosen as a function of the state alone. We use their theorem to establish the feasible-set equivalence and build on their constructive proof to describe the implementing signals. We then characterize when these signals can have a simple nested-interval structure and use the construction in two economic applications.
	
	Our construction proceeds in two steps. In the first step, we follow \cite{gentzkow2016rothschild} and \cite{Sun2026competition} by comparing the integrated cumulative distribution functions of the prior and target. Bayes plausibility requires the integrated distribution of the target to lie weakly below that of the prior. As shown by \cite{Sun2026competition}, any cutoff at which this constraint binds is a \textit{separating cutoff}: every implementing experiment must reveal whether the realized state lies above or below it. The regions where the constraint is slack form a finite or countable collection of \textit{non-separating intervals}. On each such interval, the prior and target have the same probability mass and conditional mean. We show that any deterministic implementation must fully reveal the state on the set of separating cutoffs and implement the target independently within each non-separating interval. This reduces the global implementation problem to a collection of local ones.
	
	In the second step, we construct a deterministic signal within each non-separating interval. We use the \textit{barcode transport} method of \cite{Nutz2024martingale}, which repeatedly matches pieces of the target to suitable parts of the remaining prior. Each pair has the same probability mass and conditional mean, and the prior pieces do not overlap except on sets of probability zero. Within each pair, the \textit{curtain coupling} of \cite{BeiglbockJuillet2016} provides an assignment of posterior means to states. The pieces are chosen so that almost every state receives a unique assignment. Combining these assignments yields the desired deterministic signal. We illustrate the procedure with $K$-pooling and normal distributions, showing how the construction works for both discrete and continuous target distributions. 
	
	We then characterize when a simpler implementation is available. On each non-separating interval, a lower contour set collects the states assigned posterior means no greater than a given threshold. An upper contour set collects those assigned means no less than the threshold. We call an implementation \emph{contour nested-interval} if, on every non-separating interval, either all lower contour sets or all upper contour sets are intervals of states. The direction may differ across non-separating intervals. 
	For each threshold and contour direction, matching the probability mass and first moment of the corresponding part of the target determines a candidate interval of states. A contour nested-interval implementation is feasible exactly when these intervals expand for lower contours or shrink for upper contours as the threshold rises. For a given contour direction, a feasible implementation is unique up to sets of prior probability zero and can be obtained from the corresponding curtain coupling. We also give necessary conditions based on the target's support and sufficient conditions based on the prior and target densities and cumulative distributions. These conditions help determine whether a simple contour structure is available and which direction to use.
	
	Finally, we apply our results to the two information design problems mentioned above: monopolistic pricing with buyer-optimal learning \citep{Roesler2017buyer} and robust persuasion with ambiguous receiver preferences \citep{SapiroGheiler2024}. In both settings, the designer's objective is nonlinear in the posterior-mean distribution. In the cases we study, the optimal targets have a continuum of posterior means within a non-separating interval and are therefore not bi-pooling. Existing implementation results for bi-pooling distributions \citep{Kleiner2021extreme,Arieli2023optimal} thus do not settle how to implement these targets. Our results show both that deterministic implementation is possible and that the implementing signals can have a simple nested-interval structure. 
	For buyer-optimal learning under a uniform prior, we explicitly construct an implementation of the optimal truncated equal-revenue distribution. Its two non-separating intervals require different contour directions: lower contours on the first and upper contours on the second. For robust persuasion, we consider non-uniform concave priors with a continuous, positive density and a positive mean receiver cutoff no greater than the prior mean. We show that the robustly optimal uniform target admits an upper contour nested-interval implementation. 
	
	\paragraph{Related literature}
	This paper is closely related to \cite{Arieli2025constructive}, who address the same implementation question by providing a constructive version of Blackwell's theorem. They show that every mean-preserving contraction of an atomless prior can be generated by a downward-uniform signal. Conditional on each state, their signal is drawn uniformly from an interval whose upper bound depends on that state. Our paper complements their approach by constructing implementations that use deterministic signals only. 
	
	Our feasible-set equivalence also clarifies the role of restrictions on the structure of a deterministic signal. A general deterministic signal can pool disconnected sets of states, and the posterior mean it induces need not vary monotonically with the state. Restricting attention to monotone partitions can therefore be costly. \cite{Kolotilin2024monotone} show that such a restriction can strictly reduce the designer's payoff, and \cite{Chen2025explainable} further quantify the resulting value losses. Our result shows that, with an atomless prior, allowing general partitions restores the value attainable under unrestricted experiments.
	
	More broadly, our analysis relates to the literature on mean-based persuasion and majorization. Most papers in this literature take the designer's objective as primitive and characterize the optimal posterior-mean distributions. \cite{Kleiner2021extreme} and \cite{Arieli2023optimal} characterize the extreme points of the mean-preserving contraction set as bi-pooling distributions. \cite{Arieli2023optimal} also show that every bi-pooling distribution admits a deterministic implementation with a nested-interval structure. These results provide deterministic optimal signals with a simple structure for standard linear persuasion problems.\footnote{For an example with a nonlinear objective, \cite{Chen2025information} study information design in allocation with costly verification, where a principal chooses an allocation mechanism in response to the designed information. In their single-agent setting, an optimal information structure can be chosen to pool an upper interval and reveal lower states. This is a special case of bi-pooling and admits a simple deterministic implementation.} 
	Other papers establish the optimality of nested-interval structures in richer persuasion environments. For example, \cite{Guo2019interval} obtain nested acceptance intervals for a privately informed receiver, and \cite{Kolotilin2025persuasion} study pairwise signals, single-dipped and single-peaked disclosure, and conditions under which all pooled pairs are nested. We instead take the target distribution as primitive and ask how to implement it deterministically. Our contour nested-interval characterization is stated directly in terms of the prior and target and applies beyond the bi-pooling class.
	
	Finally, our paper relates to the literature on equivalence between stochastic and deterministic mechanisms. \cite{Chen2019equivalence} use mutual purification to preserve interim expected allocation probabilities and utilities, as well as ex-ante social surplus, in a mechanism design environment with finitely many alternatives and independent atomless types. \cite{Zeng2023derandomization} establish conditions under which a stochastic persuasion mechanism with finite message and action spaces has an effectively equivalent deterministic counterpart with the same set of equilibria. Our analysis uses backward-Monge martingale transport to preserve a prescribed posterior-mean distribution exactly, including targets with a continuum of posterior means. Preserving this distribution yields value equivalence whenever the objective and constraints depend on the experiment only through that distribution. 
	
	The remainder of the paper is organized as follows. Section~\ref{Sec: Feasible_Set_Equivalence} introduces the framework and establishes feasible-set equivalence between deterministic and stochastic experiments. Section~\ref{Sec: Construction} presents the deterministic construction and illustrates it with examples. Section~\ref{Sec: Contour_nested_interval} characterizes contour nested-interval implementations and provides conditions for their feasibility. Section~\ref{Sec: Applications} applies the results to buyer-optimal learning and robust persuasion with ambiguous receiver preferences. Section~\ref{Sec} concludes, and Appendix~\ref{App: Proofs} collects the omitted proofs.

	\section{Feasible-set equivalence}\label{Sec: Feasible_Set_Equivalence}
	Let the \textit{state} $X$ be a real-valued random variable that follows an atomless distribution with finite first moment and cumulative distribution function (cdf) $F$ on $\R$. 
	
	A \emph{stochastic experiment} consists of a standard Borel message space $S$ and a stochastic mapping $\kappa: \R\mapsto\Delta(S)$ that maps each state to a distribution of messages. 
	Each message then induces a posterior mean of state, denoted by $\E[X|M]$, through Bayes' rule. 
	Without loss of generality, we impose $S=\R$ and $M=\E[X|M]$ so the message directly reflects the posterior mean state induced by Bayes' rule.
	Let $G$ denote the resulting marginal distribution of message $M$. 
	We can therefore represent an unrestricted experiment by a joint distribution $\gamma$ of the random vector $(M,X)$ that satisfies $M\sim G$, $X\sim F$, and $\E[X|M]=M$. 
	The joint distribution $\gamma$ is called a \textit{martingale coupling} of $G$ and $F$. Thus, given prior $F$, a posterior-mean distribution $G$ can be feasibly induced by some stochastic experiment if and only if there exists a martingale coupling of $G$ and $F$. 
	
	A \emph{deterministic experiment} assigns each state a single message, with no randomization conditional on the state. Again let the message space be $S=\R$ and relabel each message $M$ by the posterior mean state it induces. A deterministic signal can then be represented by a Borel mapping $T:\R\mapsto\R$ that maps each state $X$ to a message $T(X)$, with $\E[X|T(X)]=T(X)$. It can also be equivalently described by a partition of the state space, whose cells are the nonempty level sets $T^{-1}(\{m\})$ indexed by $m\in\R$. Clearly, every deterministic experiment is also a stochastic experiment, but not vice versa. Let
	\begin{equation*}
		\DET(F):=\{H:\ \exists\ \text{a Borel}\ T\ s.t.\ X\sim F,\ T(X)\sim H,\ \text{and } \E[X|T(X)]=T(X)\}
	\end{equation*}
	denote the set of posterior-mean distributions that can be induced by some deterministic experiment. 
	Given prior $F$, we say that a posterior-mean distribution $G$ is \textit{deterministically implementable} if $G\in\DET(F)$, and a deterministic experiment $T$ \textit{implements $G$ from $F$} if it satisfies $T(X)\sim G$ and $\E[X|T(X)]=T(X)$ for $X\sim F$.  
	The aim of the paper is to characterize the set of all deterministically implementable posterior-mean distributions, and for every such distribution provide an explicit construction of $T$ that implements it from the given prior. 
	
	\paragraph{Feasible-set equivalence.}
	Now we characterize the feasible sets of posterior-mean distributions under stochastic and deterministic experiments, and show that they coincide when the prior $F$ is atomless. 
	
	Let $\Delta(\R)$ denote the set of cdfs of random variables on $\R$ with finite first moment. 
	For any $G,H\in\Delta(\R)$, we say that $G$ is a \textit{mean-preserving spread} of $H$, written $H\preceq_{MPS}G$ if $\E_H[\varphi(\cdot)]\leq\E_G[\varphi(\cdot)]$ for every convex function $\varphi:\R\to\R$; in this case we also say that $H$ is a \textit{mean-preserving contraction} of $G$. 
	For any prior $F\in\Delta(\R)$, let
	\begin{equation}\label{Eq: MPC}
		\MPC(F):=\{H\in\Delta(\R):\ H\preceq_{MPS}F\}
	\end{equation}
	denote the set of mean-preserving contractions of prior $F$. 	
	
	Recall from the earlier discussion that a posterior-mean distribution $G$ can be implemented by a stochastic signal under prior $F$ if and only if there exists a martingale coupling of $G$ and $F$. 
	By \citet[Theorem~8]{Strassen1965existence}, such a coupling exists if and only if $G\preceq_{MPS}F$. 
	Therefore, $\MPC(F)$ is exactly the set of posterior-mean distributions that are feasible under unrestricted experiments. 
	This immediately implies
	\begin{equation}
		\DET(F)\subseteq\MPC(F)
	\end{equation}
	for any prior $F$, that is, every posterior-mean distribution that can be induced by a deterministic experiment must also be a mean-preserving contraction of $F$. 
	This is because the set of feasible experiments is larger with stochastic experiments. So, any outcome that can be implemented with deterministic experiments can also be done with stochastic ones.
	
	The more interesting question is the converse one: can every $G\in\MPC(F)$ be implemented by some deterministic experiment? 
	Strassen's theorem guarantees that for any such $G$ there exists a joint distribution $\gamma$ of random vector $(M,X)$ such that $M\sim G$, $X\sim F$ and $\E[X|M]=M$. 
	However, deterministic implementation in addition requires $M$ to be a measurable function of state $X$; or equivalently, $M=T(X)$ for some Borel mapping $T$. 
	A martingale coupling $\gamma$ of $G$ and $F$ that also satisfies this additional property is called \textit{backward Monge} \citep{Nutz2024martingale}. 
	The following result, due to \cite{Nutz2024martingale}, shows that if $F$ is atomless then for any $G\in\MPC(F)$ there exists a backward Monge martingale coupling of $G$ and $F$. 
	
	\begin{lemma}
		\label{Lemma: NWZ_backward_Monge}
		Let both $F$ and $G$ be cdfs of random variables on $\R$ with finite first moment, and let $F$ be atomless. 
		If $G\preceq_{MPS}F$, then there exists a Borel mapping $T:\R\to\R$ such that $X\sim F$, $T(X)\sim G$, and $\E[X|T(X)]=T(X)$.\footnote{This lemma is an atomless specialization of the backward-Monge existence theorem of \citet[Theorem~2.1]{Nutz2024martingale}, stated using our terminology.}  
	\end{lemma}
	
	This delivers our feasible-set equivalence result: 
	\begin{theorem}[Feasible-set equivalence]\label{Theorem: Feasible_Set_Equivalence}
		If the prior $F$ is atomless and has a finite first moment, then $\DET(F)=\MPC(F)$. 
		That is, for every target distribution $G\in\MPC(F)$, there exists a deterministic experiment $T$ that implements $G$ from $F$. 
	\end{theorem}
	\begin{proof}[Proof of Theorem~\ref{Theorem: Feasible_Set_Equivalence}]
		The inclusion $\DET(F)\subseteq\MPC(F)$ follows from the fact that $\MPC(F)$ is the feasible set of posterior-mean distributions for all stochastic experiments, which contain deterministic experiments as special cases. 
		To show $\MPC(F)\subseteq\DET(F)$, consider any $G\in\MPC(F)$. 
		Since $F$ is atomless and $G\preceq_{MPS}F$, Lemma~\ref{Lemma: NWZ_backward_Monge} implies that there exists a Borel mapping $T$ with $T(X)\sim G$ and $\E[X|T(X)]=T(X)$ for $X\sim F$. Using $T(X)$ as the message therefore generates posterior mean $T(X)$ with cdf $G$. Therefore, $G\in\DET(F)$. 
	\end{proof}
	
	\textbf{Implication of Theorem~\ref{Theorem: Feasible_Set_Equivalence}.} 
	As mentioned in the Introduction, in many economically relevant applications the designer's reduced-form problem is 
	\begin{equation*}
		\max_{G}{V(G)} \qquad \text{subject to} \qquad G\in\mathcal{A}\subseteq\MPC(F)
	\end{equation*}
	where $\mathcal{A}$ is the feasible set of distributions that are allowed. 
	Our feasible-set equivalence result implies that for any $G\in\mathcal{A}$ there exists a deterministic signal that can implement it. Hence, the value of this problem is the same whether the designer can use unrestricted stochastic signals or only deterministic signals. 
	This applies to any objective and constraints that depend on the experiment only through the posterior-mean distribution. 
	
	\textbf{Remark}. 
	The atomlessness assumption of prior distribution $F$ is crucial for the general feasible-set equivalence theorem. To see this, consider $\mu_F=\tfrac{1}{2}\delta_0+\tfrac{1}{2}\delta_1$ and $\mu_G=\tfrac{1}{2}\delta_{1/4}+\tfrac{1}{2}\delta_{3/4}$, where $\delta_a$ denotes unit mass at $a$. Both distributions have mean $1/2$, and $G\in\MPC(F)$. 
	A deterministic signal, however, can either pool the two states $0$ and $1$ into a single message or separate them by different messages. 
	In the first case, the posterior-mean law is degenerate at $1/2$; in the second, it is $\mu_F$. 
	Neither equals $\mu_G$, so $G\notin\DET(F)$.

	\section{How to deterministically implement a target distribution?}\label{Sec: Construction}
	
	Consider an atomless prior $F$ and fix a target distribution $G\in\MPC(F)$. 
	While Theorem~\ref{Theorem: Feasible_Set_Equivalence} guarantees the existence of a deterministic experiment that implements $G$ from $F$, it is silent about how to explicitly construct such an experiment. 
	In this section we present a constructive way to deduce a deterministic mapping $T$ that implements $G$ from $F$. 
	
	The construction contains two steps. 
	In the first step, we follow \cite{Sun2026competition} by decomposing $\R$ into a set of \textit{separating cutoffs} and a collection of \textit{non-separating intervals}, using the majorization constraints governing the target $G$ and prior $F$ \citep{gentzkow2016rothschild}. 
	On the set of separating cutoffs, $G$ and $F$ match exactly, and any deterministic experiment $T$ that implements $G$ from $F$ must be fully revealing on this set. 
	On each non-separating interval, the target distributed $G$ is dominated by prior $F$ in convex when restricted to that interval. 
	Any $T$ that implements $G$ from $F$ must also do the same independently within each such interval.
	Then, in the second step, we zoom in to each such non-separating interval and exploit the ideas of \textit{left-curtain} and \textit{barcode transport} from \cite{BeiglbockJuillet2016} and \cite{Nutz2024martingale} to construct a deterministic mapping that implements $G$ from $F$ when restricted to this interval only. 
	These pieces can then be glued together to obtain a deterministic mapping $T$ that implements $G$ from $F$.  
	
	\subsection{Decomposition into separating cutoffs and non-separating intervals}\label{Subsec: Decomposition}
	Following \cite{gentzkow2016rothschild}, for any cdf $H$ with finite first moment on $\mathbb{R}$ we define its \textit{integrated cdf} by 
	\begin{equation*}
		I_H(z):=\int_{-\infty}^{z}H(y)\,dy \quad \text{for all } z\in\R. 
	\end{equation*}
	For cdfs $F$ and $G$ with the same finite mean, $G\in\MPC(F)$ if and only if the following majorization constraint holds: 
	\begin{equation}
		I_F(z)-I_G(z)\geq0\ \text{for every }z\in\R\ \text{and} \lim\limits_{z\rightarrow\pm\infty}{(I_F(z)-I_G(z))}=0. 
	\end{equation}
	Thus, for any fixed target $G\in\MPC(F)$, we can define 
	\begin{equation}\label{Eq: Decomposed_Sets}
		\mathbb{S}_G:=\{z\in\R:\ I_G(z)=I_F(z)\}
		\qquad \text{and} \qquad
		\mathbb{NS}_G:=\{z\in\R:\ I_G(z)<I_F(z)\}.
	\end{equation}
	Since $\mathbb{S}_G$ is closed, its complement $\mathbb{NS}_G$ must be open so that it is either empty or a finite or countable union of disjoint intervals; that is, we can write $\mathbb{NS}_G=\bigcup_{j\in J}(c_j,d_j)$ with $J$ a finite or countable index set.\footnote{A non-separating interval can be unbounded on either or both ends if the support of prior $F$ is unbounded on $\R$. Moreover, $\mathbb{NS}_G$ is empty if and only if $G=F$. }
	
	\cite{Sun2026competition} shows that both $\mathbb{S}_G$ and every interval $(c_j,d_j)$ in $\mathbb{NS}_G$ have economically appealing interpretations. To understand them, let $H_z$ denote the posterior-mean distribution induced by a binary cutoff signal that only reveals whether the realize state is above or below cutoff $z$, given prior $F$. \citet[Lemma~2]{Sun2026competition} show that $z\in\mathbb{S}_G$ if and only if $H_z\preceq_{MPS}G$, that is, any (stochastic) experiment that implements $G$ from $F$ must reveal whether the realized state lies above or below cutoff $z$. Every $z\in\mathbb{S}_G$ can be thus interpreted as a \textit{separating cutoff}. By contrast, the intervals in $\mathbb{NS}_G$ are regions in which such clean separation does not hold. Hence, each $(c_j,d_j)$ can be interpreted as a \textit{non-separating interval}. \cite{Sun2026competition} also shows that, on each $(c_j,d_j)$, the prior $F$ and target $G$ have the same mass and mean, and $G$ is a mean-preserving contraction of $F$ when restricted to this interval. That is, $G|_{(c_j,d_j)}\preceq_{MPS}F|_{(c_j,d_j)}$, where $G|_{(c_j,d_j)}$ and $F|_{(c_j,d_j)}$ denote distributions $G$ and $F$ truncated on $(c_j,d_j)$.\footnote{Formally, for any interval $(c,d)$ with $G(c)<G(d^{-})$, $G(x)|_{(c,d)}=\frac{G(x)-G(c)}{G(d^{-})-G(c)}$ for all $x\in(c,d)$, $G(x)|_{(c,d)}=0$ for $x\leq c$ and $G(x)|_{(c,d)}=1$ for $x\geq d$. The construction of $F|_{(c,d)}$ is similar.}
	
	Proposition~\ref{Prop: Decomposition} clarifies the implications of these sets for any deterministic experiment $T$ that implements $G$ from $F$, whose existence is guaranteed by Theorem~\ref{Theorem: Feasible_Set_Equivalence} since $F$ is atomless.

	\begin{proposition}\label{Prop: Decomposition}
		Suppose that $F$ is atomless and $G\in\MPC(F)$ with $\mathbb{NS}_G=\cup_{j\in J}(c_j,d_j)$. 
		Let $T$ be a deterministic mapping that implements $G$ from $F$. 
		Then the following properties hold: 
		\begin{enumerate}
			\item $T(x)=x$ for a.e. $x\in\mathbb{S}_G$; that is, $T$ is fully revealing on the set of separating cutoffs. 
			\item On each $(c_j,d_j)$, $G|_{(c_j,d_j)}\preceq_{MPS}F|_{(c_j,d_j)}$ and $T$ implements $G|_{(c_j,d_j)}$ from $F|_{(c_j,d_j)}$. 
		\end{enumerate}
	\end{proposition}
	\begin{proof}[Proof of Proposition~\ref{Prop: Decomposition}]
		See Appendix~\ref{App. Prop_Decomposition}. 
	\end{proof}
	
	Proposition~\ref{Prop: Decomposition} thus simplifies the construction of a deterministic experiment $T$ implementing $G$ from $F$ by decomposing the global problem into independent implementation problems on each of the non-separating intervals of $G$. 
	The set of separating cutoffs must be fully revealed by any such $T$.

	\subsection{Barcode transport within each non-separating interval}\label{Subsec: Barcode_Transport}
	In this subsection we introduce a specific way, based on \cite{BeiglbockJuillet2016} and \cite{Nutz2024martingale}, to explicitly construct a deterministic mapping that implements $G|_{(c_j,d_j)}$ from $F|_{(c_j,d_j)}$ on each non-separating interval $(c_j,d_j)$ of $G$. 
	These pieces can then be merged together to yield the ``grand'' $T$ that implements $G$ from an atomless $F$. 
	
	\paragraph{Convex order, extended convex order, and shadow}
	We begin by introducing some notation and terminology. 
	Let $\mathcal{M}$ denote the set of finite Borel measures on $\R$ with finite first moment, endowed with the usual topology of weak convergence.\footnote{The reason we work with $\mathcal{M}$ instead of the set of Borel probability measures on $\R$ is that we will need to consider transport plans between measures in $\mathcal{M}$ that have the same mass but not necessarily equal to one.}  
	Let $\mu,\nu\in\mathcal{M}$. We say that $\nu$ dominates $\mu$ in \textit{convex order}, denoted by $\mu\leq_{cx}\nu$, if $\int_{\R}{\varphi(x)d\mu(x)}\leq\int_{\R}{\varphi(x)d\nu(x)}$ for every convex function $\varphi:\R\mapsto\R$. 
	Note that when $\mu$ and $\nu$ are probability measures (i.e., each has total mass of one) $\leq_{cx}$ is equivalent to $\preceq_{MPS}$ with respect to their cdfs.
	We say, following \cite{BeiglbockJuillet2016}, that $\nu$ dominates $\mu$ in \textit{extended convex order}, denoted by $\mu\leq_{ecx}\nu$, if $\int_{\R}{\varphi(x)d\mu(x)}\leq\int_{\R}{\varphi(x)d\nu(x)}$ for every non-negative convex function $\varphi:\mathbb{R}\mapsto\mathbb{R}_{\geq0}$. 
	When $\mu$ and $\nu$ have the same total mass, $\leq_{ecx}$ coincides with $\leq_{cx}$. But $\mu\leq_{ecx}\nu$ also holds if $\mu\leq\nu$, that is, $\mu(A)\leq\nu(A)$ for every Borel subset $A$ of $\R$. 
	Given $\mu,\nu\in\mathcal{M}$ with $\mu\leq_{ecx}\nu$, \cite{BeiglbockJuillet2016} define the \emph{shadow} of $\mu$ in $\nu$ as
	\begin{equation}\label{Eq: Shadow}
		S^\nu(\mu):=\min\left\{\eta:\mu\leq_{cx}\eta\leq\nu\right\},
	\end{equation}
	where the minimum is taken in convex order. 
	Intuitively, the shadow is the least dispersed part of $\nu$ capable of absorbing $\mu$. 
	In particular, if $\nu$ is atomless and $\mu=p\delta_m$ is an atom of mass $p$ and barycenter $m$, then the shadow of $\mu$ on $\nu$ is an interval restriction $\nu|_I$, where the interval $I$ can be characterized by the mass and first moment equalities $\nu(I)=p$ and $\int_{I}{xd\nu(x)}=pm$. 
	Another useful fact about shadow is that $S^{\nu}(\mu)=\nu$ for $\mu\leq_{cx}\nu$.

	\paragraph{Barcode transport}
	Fix an atomless prior $F$, a feasible target distribution $G$, and one of its non-separating intervals $(c_j,d_j)$. Recall that $G$ and $F$ have the same mass and conditional mean on this interval, and $G|_{(c_j,d_j)}\preceq_{MPS}F|_{(c_j,d_j)}$ (cf. Proposition~\ref{Prop: Decomposition}). 
	The \textit{barcode transport} method of \cite{Nutz2024martingale} provides a construction of a deterministic experiment implementing $G|_{(c_j,d_j)}$ from $F|_{(c_j,d_j)}$. 
	Let $\mu$ and $\nu$ denote their respective probability measures.\footnote{They can be deduced via $\mu((-\infty,x])=G(x)|_{(c_j,d_j)}$ and $\nu((-\infty,x])=F(x)|_{(c_j,d_j)}$ for all $x\in\R$.} 
	Then $\mu\leq_{cx}\nu$, and $\nu$ is atomless (i.e., $\nu(\{x\})=0$ for all $x\in\R$). 
	
	The basic idea of barcode transport is to decompose $\nu$ and $\mu$ into countably many paired pieces --- the \textit{bars} of the barcode --- that can each be coupled by a backward-Monge martingale transport, and then combine them together. 
	The construction consists of two steps. 
	In the first step, we recursively construct the bars.
	We start by identifying
	\begin{equation*}
		C_1 = \left\{x:\ \frac{d\mu}{d(\mu+\nu)}(x)\geq\frac{1}{2}\right\}
	\end{equation*}
	where the derivative here is the Radon-Nikodym density relative to $\mu+\nu$. 
	Intuitively, these are points where either $\mu$ has a weakly higher density than $\nu$ or where $\mu$ contains an atom. 
	Note that $C_1$ is not empty because $\mu$ is a mean-preserving contraction of $\nu$ so that $\mu$ must contain regions where it is more densely distributed than $\nu$. 
	Construct $\mu_1=\mu|_{C_1}$ and its shadow in $\nu$, denoted by $\nu_1=S^{\nu}(\mu_1)$.\footnote{For any Borel measure $\mu$ on $\R$ and Borel set $A\subseteq\R$, we define $\mu|_A:=\mu\cdot\1_{A}$.} 
	This yields the first pair of bars $(\mu_1,\nu_1)$. 
	More generally, for $n\geq2$, define 
	\begin{equation*}
		\tilde{\mu}_n=\mu-\sum_{i=1}^{n-1}{\mu_i} \qquad \text{and} \qquad \tilde{\nu}_n=\nu-\sum_{i=1}^{n-1}{\nu_i}\ .
	\end{equation*}
	They are the residual measures of $\mu$ and $\nu$ after removing the bars generated in the previous $n-1$ rounds. 
	We then analogously identify 
	\begin{equation*}
		C_n=\left\{x:\ \frac{d\tilde{\mu}_n}{d(\tilde{\mu}_n+\tilde{\nu}_n)}(x)\geq\frac{1}{2}\right\}\ , 
	\end{equation*}
	the set of points where $\tilde{\mu}_n$ is more densely distributed than $\tilde{\nu}_n$, and construct $\mu_n=\tilde{\mu}_n|_{C_n}$ and its shadow $\nu_n=S^{\tilde{\nu}_n}(\mu_n)$, yielding the $n$-th pair of bars $(\mu_n,\nu_n)$. 
	Repeating this procedure produces a finite or countable sequence $\{(\mu_n,\nu_n)\}_n$. 
	Each such pair $(\mu_n,\nu_n)$ satisfies $\mu_n\leq_{cx}\nu_n$ and the density dominance condition $d\mu_n/d(\mu_n+d\nu_n)\geq1/2$ $\mu_n$-a.e. 
	\cite{Nutz2024martingale} show that the target bars obtained in this way are mutually singular, so are the prior bars. 
	Moreover, they together exhaust the original measures $\mu$ and $\nu$ (that is, $\mu=\sum_{n}\mu_n$ and $\nu=\sum_{n}\nu_n$). 
	
	In the second step, we construct a deterministic mapping for each pair of bars $(\mu_n,\nu_n)$ using the \textit{left-curtain} transport of \cite{BeiglbockJuillet2016}. Given $\mu_n\leq_{cx}\nu_n$, this is the unique martingale coupling that transports $\mu_n|_{(-\infty,s]}$ to its shadow $S^{\nu_n}(\mu_n|_{(-\infty,s]})$ for every $s\in\R$.
	This describes the matching from posterior means to states: mass associated with a posterior mean is distributed across states while preserving that mean. Unless the two marginals coincide, this forward allocation necessarily involves mass splitting. 
	For implementation, we read the same matching backward, assigning a posterior-mean label to each realized state. 
	This assignment need not be deterministic for arbitrary marginals. On each barcode bar, however, $\nu_n$ is atomless and
	\begin{equation*}
		\frac{d\mu_n}{d(\mu_n+\nu_n)}\geq\frac12
		\quad\mu_n\text{-a.e.}
	\end{equation*}
	\citet[Proposition~2.2]{Nutz2024martingale} then ensures that almost every state receives a unique label. Consequently, there exists a Borel mapping $T_n:\R\to\R$ satisfying
	\begin{equation*}
		(T_n)_\#\nu_n=\mu_n \qquad \text{and} \qquad \E[X|T_n(X)]=T_n(X).
	\end{equation*}
	Here $(T_n)_\#\nu_n$ is the pushforward measure defined by $(T_n)_\#\nu_n(A)=\nu_n(T_n^{-1}(A))$ for all Borel sets $A\subseteq\R$, and the conditional expectation is taken under the normalized prior bar $\nu_n/\nu_n(\R)$. Thus $T_n$ implements $\mu_n$ from $\nu_n$. 
	Processing target means from high to low instead gives the analogous \textit{right-curtain} construction.
	
	Finally, because the collection of prior bars $\{\nu_n\}_n$ are mutual singular and $\nu=\sum_{n}\nu_n$, the non-separating interval $(c_j,d_j)$ can be partitioned into disjoint Borel sets $B_n$'s with $\nu_n=\nu|_{B_n}$ for each $n$. On each such component $B_n$, define $T(x)=T_n(x)$ for $x\in B_n$. Since the bars exhaust both measures $\mu$ and $\nu$, the resulting mapping thus has target distribution $\mu$ and satisfies $\mathbb{E}[X|T(X)]=T(X)$ on $(c_j,d_j)$. 
	Constructing these mappings on every non-separating interval and setting $T(x)=x$ on $\mathbb{S}_G$ yields a deterministic experiment that implements target $G$ from prior $F$. 
	
	\subsection{Two examples}\label{Subsec: Examples_Construction}
	We illustrate the construction of deterministic implementations with two examples. 
	
	\begin{example}[\textbf{$K$-pooling distributions}]\label{Example: K-pooling}
		\normalfont
		For an integer $K\geq1$, we call $G\in\MPC(F)$ a \emph{$K$-pooling distribution} if its restriction to every non-separating interval has at most $K$ atoms and no other mass. Outside these intervals, the state is fully revealed. The cases $K=1$ and $K=2$ correspond, respectively, to monotone partitional distributions and bi-pooling distributions, which are studied in the literature \citep{Dworczak2019,Kleiner2021extreme,Arieli2023optimal}. 
		Note that a $K$-pooling distribution need not have finite support globally; it may have countably many non-separating intervals and a fully revealing part.
		
		Fix one component $(c_j,d_j)$, and let $\mu$ and $\nu$ denote the probability measures for $G|_{(c_j,d_j)}$ and $F|_{(c_j,d_j)}$, respectively. 
		Suppose that the target measure $\mu$ has finite support with $K_j\leq K$ atoms: 
		\begin{equation*}
			\mu=\sum_{k=1}^{K_j}p_k\delta_{x_k},
			\qquad
			c_j<x_1<\cdots<x_{K_j}<d_j,
			\qquad
			p_k>0,\quad \sum_{k=1}^{K_j}p_k=1. 
		\end{equation*}
		Since $\nu$ is atomless, we have
		\begin{equation*}
			\frac{d\mu}{d(\mu+\nu)}(x) = 
			\begin{cases}
				1, &\text{if } x\in\{x_1,\cdots,x_{K_j}\}, \\
				0, &\text{otherwise.}
			\end{cases}
		\end{equation*}
		Therefore, the set of points where $\mu$ is distributed more densely than $\nu$ is precisely $C_1=\{x_1,\cdots,x_{K_j}\}$, the atoms of $\mu$ on $(c_j,d_j)$. 
		Then $\mu_1=\mu|_{C_1}=\mu$ and its shadow is $\nu_1=S^\nu(\mu_1)=S^\nu(\mu)=\nu$ since $\mu\leq_{cx}\nu$. 
		The first step of barcode transport thus consists of only one round and it produces a single pair of bars $(\mu,\nu)$. 
		
		The second step constructs a deterministic mapping within this single bar. Consider the left-curtain transport first.
		As the target mean level $s$ increases, $\mu|_{(-\infty,s]}$ starts at zero, gains the atom $p_k\delta_{x_k}$ at each $x_k$, and reaches $\mu$ at $x_{K_j}$.
		These $K_j$ jumps correspond to $K_j$ rounds of construction.
		All interval and cell identities below are understood within $(c_j,d_j)$, up to $\nu$-null sets.
		
		In the first round, choose the interval $I_1^L=[a_1^L,b_1^L]$ such that the restriction of $\nu$ to it is the shadow of $p_1\delta_{x_1}$:
		\begin{equation*}
			S^\nu(p_1\delta_{x_1})=\nu|_{I_1^L},
			\qquad
			\nu(I_1^L)=p_1,
			\qquad \text{and} \qquad
			\int_{I_1^L}x\,d\nu(x)=p_1x_1.
		\end{equation*}
		Set $D_1^L=I_1^L$ and assign $T_L(x)=x_1$ for $x\in D_1^L$. 
		Its mass is $p_1$ and its conditional mean is $x_1$.
		Removing it leaves the residual prior measure $\rho_1=\nu-\nu|_{D_1^L}$. 
		
		If $K_j\geq2$, the second round matches $p_2\delta_{x_2}$ to its shadow in $\rho_1$. 
		Choose an interval $I_2^L=[a_2^L,b_2^L]$ satisfying 
		\begin{equation*}
			\rho_1(I_2^L)=p_2 \qquad \text{and} \qquad \int_{I_2^L}x\,d\rho_1(x)=p_2x_2\ . 
		\end{equation*}
		Let $D_2^L=I_2^L\setminus D_1^L$, then 
		\begin{equation*}
			S^{\rho_1}(p_2\delta_{x_2})=\rho_1|_{I_2^L}=\nu|_{D_2^L}. 
		\end{equation*} 
		Assign $T_L(x)=x_2$ on $D_2^L$.
		The distinction between $I_2^L$ and $D_2^L$ matters: the former is an interval used to restrict the remaining prior, while the latter excludes states already assigned in the first round.
		
		More generally, starting from $\rho_0=\nu$, repeat this procedure recursively for $k=1,\ldots,K_j$:
		\begin{equation*}
			S^{\rho_{k-1}}(p_k\delta_{x_k})=\rho_{k-1}|_{I_k^L}=\nu|_{D_k^L},
			\qquad
			D_k^L=I_k^L\setminus\bigcup_{i=1}^{k-1}D_i^L,
			\qquad
			\rho_k=\rho_{k-1}-\nu|_{D_k^L}.
		\end{equation*}
		At each round $k$, $I_k^L$ is chosen to preserve the atom's mass $p_k$ and mean $x_k$ under the residual prior $\rho_{k-1}$.
		Shadow associativity \citep[][Theorem~4.8]{BeiglbockJuillet2016} gives
		\begin{equation*}
			\nu|_{\bigcup_{i=1}^{k}D_i^L} = S^\nu\!\left(\sum_{i=1}^{k}p_i\delta_{x_i}\right).
		\end{equation*}
		Thus, the construction satisfies the cumulative-shadow characterization of left-curtain transport. 
		At $k=K_j$, the right-hand side is $S^\nu(\mu)=\nu$, so the disjoint cells exhaust the prior. Assigning $T_L(x)=x_k$ on each $D_k^L$ then implements $\mu$ from $\nu$; any remaining null states can be assigned arbitrarily. 
		Similarly, the right-curtain transport processes the atoms from high to low and constructs a Borel mapping $T_R$ implementing $\mu$ from $\nu$, together with its partition cells $D_k^R$. 
		
		Repeating this construction on every non-separating interval and setting $T(x)=x$ on $\mathbb S_G$ implements the entire $K$-pooling distribution. The number of posterior-mean cells is at most $K$ on each component, although their locations and probabilities can differ across components.
		
		When $K=1$, every non-separating interval has a one-point local target $\mu=\delta_{x_1}$, where $x_1$ is the prior's conditional mean on that interval. 
		The entire component must therefore be pooled at $x_1$: $T_L(x)=T_R(x)=x_1$ for all $x\in(c_j,d_j)$. 
		This local implementation is unique up to $\nu$-null sets. 
		Together with full revelation on $\mathbb S_G$, it gives a unique monotone partitional experiment under both curtain orders. 
		
		When $K=2$, the target is a \textit{bi-pooling} distribution. A component with $K_j=1$ is handled as above. On a component with $K_j=2$, the two curtain orders may give different mappings. 
		In this case, either left- or right-curtain transport construction assigns an interval to the first posterior mean it processes and assigns the remaining states to the other posterior mean. 
		Consequently, within this component, the lower contour sets of $T_L$ are nested intervals, whereas the upper contour sets of $T_R$ are nested intervals. 
		This recovers the familiar nested structure of bi-pooling implementations established by \citet[][Lemma~4]{Arieli2023optimal}.
		For example, let $\nu$ be uniform on $(0,1)$ and the target measure be $\mu=\frac12\delta_{1/3}+\frac12\delta_{2/3}$. 
		The left-curtain construction first pools $D_1^L=(1/12,7/12)$ at $x_1=1/3$, and pools its complement at $x_2=2/3$.
		The right-curtain construction instead first pools states in $D_2^R=(5/12,11/12)$ at $x_2=2/3$, and pools its complement at $x_1=1/3$. 
		Figure~\ref{Fig: Bi-pooling} compares these two implementations. Although the mappings differ, both have the corresponding nested interval contours.
		
		\begingroup
		\tikzset{
			curtain blue/.style={
				draw=black,line width=0.4pt,
				pattern=north west lines,pattern color=blue!65!black
			},
			curtain orange/.style={
				draw=black,line width=0.4pt,
				pattern=north east lines,pattern color=orange!85!black
			},
			curtain green/.style={
				draw=black,line width=0.4pt,
				pattern=north west lines,pattern color=green!45!black
			}
		}
		\begin{figure}[htbp]
			\centering
			\caption{Barcode transport implementations of a bi-pooling distribution with $K=2$}
			\begin{subfigure}[t]{0.47\textwidth}
				\centering
				\subcaption{via left-curtain transport}
				\begin{tikzpicture}[x=5.15cm,y=0.68cm]
					\draw[->,line width=0.55pt]
					(-0.02,6.05)--(1.04,6.05)
					node[right,font=\scriptsize] {$m$};
					\node[left,font=\scriptsize]
					at (-0.02,6.05) {$\mu$};
					
					\draw[curtain blue]
					(0.318,6.05) rectangle (0.362,6.70);
					\draw[curtain orange]
					(0.678,6.05) rectangle (0.722,6.70);
					
					\node[above,font=\scriptsize,text=blue!65!black]
					at (0.34,6.70) {$x_1$};
					\node[above,font=\scriptsize,text=orange!85!black]
					at (0.70,6.70) {$x_2$};
					
					\draw[curtain orange]
					(0,4.15) rectangle (0.0833,4.43);
					\draw[curtain blue]
					(0.0833,4.15) rectangle (0.5833,4.43);
					\draw[curtain orange]
					(0.5833,4.15) rectangle (1,4.43);
					
					\draw[->,line width=0.55pt]
					(-0.02,4.15)--(1.04,4.15)
					node[right,font=\scriptsize] {$x$};
					\node[left,font=\scriptsize]
					at (-0.02,4.15) {$\nu$};
					\node[above,font=\scriptsize,text=blue!65!black]
					at (0.333,4.34) {$D_1^L$};
					\node[above,font=\scriptsize,text=orange!85!black]
					at (0.80,4.34) {$D_2^L$};
					
					\draw[->,line width=0.55pt]
					(-0.035,0.25)--(1.04,0.25)
					node[right,font=\scriptsize] {$x$};
					\draw[->,line width=0.55pt]
					(-0.035,0.25)--(-0.035,2.85);
					\node[anchor=west,font=\scriptsize]
					at (0.04,2.62) {$T_L(x)$};
					
					\draw[blue!65!black,line width=1.2pt,line cap=butt]
					(0.0833,1.15)--(0.5833,1.15);
					\draw[orange!85!black,line width=1.2pt,line cap=butt]
					(0,2.10)--(0.0833,2.10);
					\draw[orange!85!black,line width=1.2pt,line cap=butt]
					(0.5833,2.10)--(1,2.10);
					
					\node[left,font=\scriptsize]
					at (-0.045,1.15) {$x_1$};
					\node[left,font=\scriptsize]
					at (-0.045,2.10) {$x_2$};
					\node[below,font=\scriptsize]
					at (0,0.25) {$c_j$};
					\node[below,font=\scriptsize]
					at (1,0.25) {$d_j$};
				\end{tikzpicture}
			\end{subfigure}\hfill
			\begin{subfigure}[t]{0.47\textwidth}
				\centering
				\subcaption{via right-curtain transport}
				\begin{tikzpicture}[x=5.15cm,y=0.68cm]
					\draw[->,line width=0.55pt]
					(-0.02,6.05)--(1.04,6.05)
					node[right,font=\scriptsize] {$m$};
					\node[left,font=\scriptsize]
					at (-0.02,6.05) {$\mu$};
					
					\draw[curtain blue]
					(0.318,6.05) rectangle (0.362,6.70);
					\draw[curtain orange]
					(0.678,6.05) rectangle (0.722,6.70);
					
					\node[above,font=\scriptsize,text=blue!65!black]
					at (0.34,6.70) {$x_1$};
					\node[above,font=\scriptsize,text=orange!85!black]
					at (0.70,6.70) {$x_2$};
					
					\draw[curtain blue]
					(0,4.15) rectangle (0.4167,4.43);
					\draw[curtain orange]
					(0.4167,4.15) rectangle (0.9167,4.43);
					\draw[curtain blue]
					(0.9167,4.15) rectangle (1,4.43);
					
					\draw[->,line width=0.55pt]
					(-0.02,4.15)--(1.04,4.15)
					node[right,font=\scriptsize] {$x$};
					\node[left,font=\scriptsize]
					at (-0.02,4.15) {$\nu$};
					\node[above,font=\scriptsize,text=blue!65!black]
					at (0.20,4.34) {$D_1^R$};
					\node[above,font=\scriptsize,text=orange!85!black]
					at (0.667,4.34) {$D_2^R$};
					
					\draw[->,line width=0.55pt]
					(-0.035,0.25)--(1.04,0.25)
					node[right,font=\scriptsize] {$x$};
					\draw[->,line width=0.55pt]
					(-0.035,0.25)--(-0.035,2.85);
					\node[anchor=west,font=\scriptsize]
					at (0.04,2.62) {$T_R(x)$};
					
					\draw[blue!65!black,line width=1.2pt,line cap=butt]
					(0,1.15)--(0.4167,1.15);
					\draw[blue!65!black,line width=1.2pt,line cap=butt]
					(0.9167,1.15)--(1,1.15);
					\draw[orange!85!black,line width=1.2pt,line cap=butt]
					(0.4167,2.10)--(0.9167,2.10);
					
					\node[left,font=\scriptsize]
					at (-0.045,1.15) {$x_1$};
					\node[left,font=\scriptsize]
					at (-0.045,2.10) {$x_2$};
					\node[below,font=\scriptsize]
					at (0,0.25) {$c_j$};
					\node[below,font=\scriptsize]
					at (1,0.25) {$d_j$};
				\end{tikzpicture}
			\end{subfigure}
			\label{Fig: Bi-pooling}
		\end{figure}
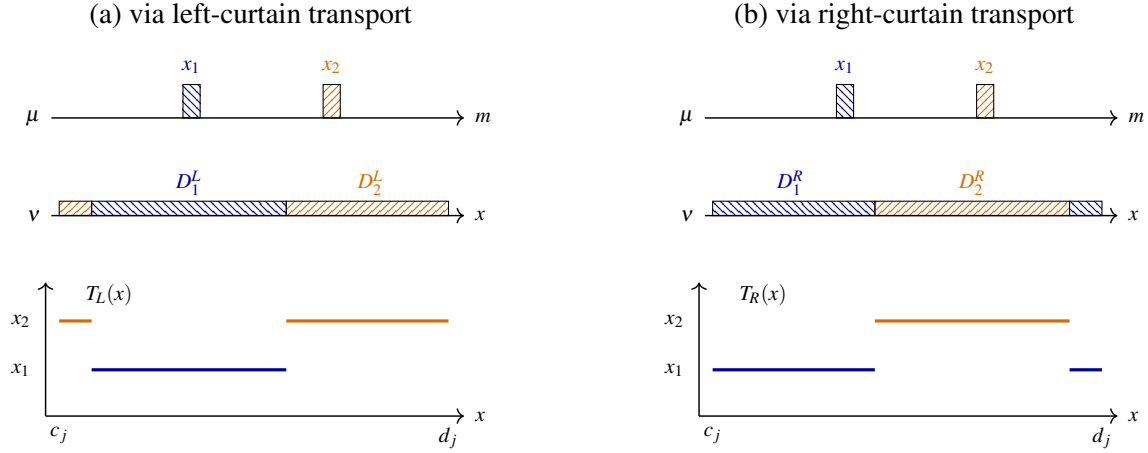
		\endgroup
		
		When $K\geq3$, this nested-interval structure is no longer guaranteed. 
		To illustrate this, again let $\nu$ be uniform on $(0,1)$ and consider a three-point target measure $\mu=\frac15\delta_{1/5}+\frac15\delta_{1/2}+\frac35\delta_{3/5}$. 
		The left-curtain transport first pools an interval at $x_1=1/5$, then pools another interval from the remaining prior at $x_2=1/2$, and finally assigns the residual to $x_3=3/5$. 
		The resulting cells are
		\begin{equation*}
			D_1^L=(1/10,3/10), \quad D_2^L=(2/5,3/5), \quad \text{and} \quad D_3^L=(0,1/10)\cup(3/10,2/5)\cup(3/5,1)
		\end{equation*}
		The first two cells are separated, so the lower contour set $\{x:\ T_L(x)\leq x_2\}=D_1^L\cup D_2^L$ is not an interval (cf. Figure~\ref{Fig: Three-point-distribution}a). 
		The right-curtain transport processes the same posterior means in the opposite order. 
		It first pools an interval at $x_3=3/5$, constructs the cell for $x_2=1/2$ from the remaining prior, and assigns the residual to $x_1=1/5$ (cf. Figure~\ref{Fig: Three-point-distribution}b). This gives 
		\begin{equation*}
			D_3^R=(3/10,9/10), \quad D_2^R=(7/40,3/10)\cup(9/10,39/40), \quad \text{and} \quad D_1^R=(0,7/40)\cup(39/40,1).
		\end{equation*}
		Here the second cell extends the first on both sides, so the upper contour sets remain nested intervals: $\{x:\ T_R(x)\geq x_3\}=(3/10,9/10)$ and $\{x:\ T_R(x)\geq x_2\}=(7/40,39/40)$. 
		Figure~\ref{Fig: Three-point-distribution} illustrates this difference. Both mappings implement the same target, but the nested interval structure now depends on the direction of construction.
		\qed
		
		\begingroup
		\tikzset{
			curtain blue/.style={
				draw=black,line width=0.4pt,
				pattern=north west lines,pattern color=blue!65!black
			},
			curtain orange/.style={
				draw=black,line width=0.4pt,
				pattern=north east lines,pattern color=orange!85!black
			},
			curtain green/.style={
				draw=black,line width=0.4pt,
				pattern=north west lines,pattern color=green!45!black
			}
		}
		
		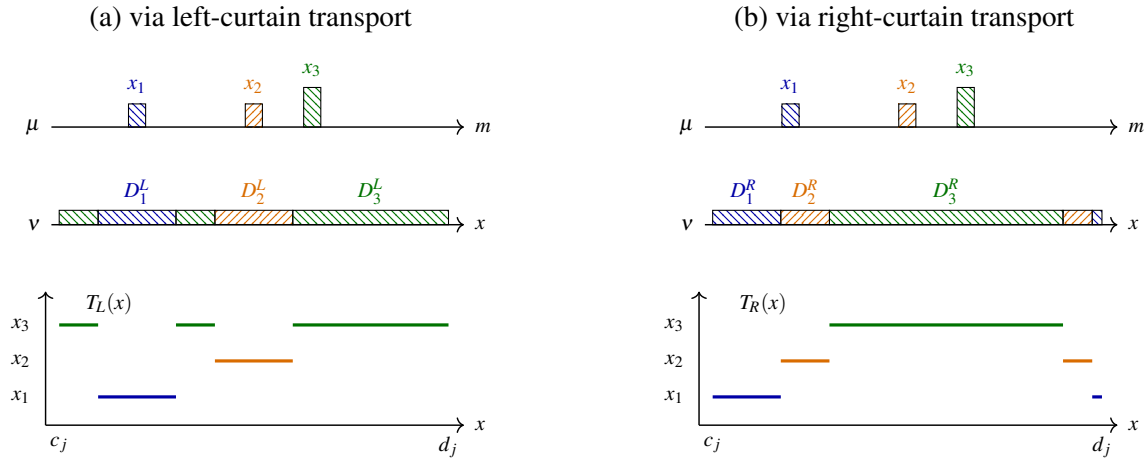
\begin{figure}[htbp]
			\centering
			\caption{Barcode transport implementations of a three-point distribution with $K=3$}
			\begin{subfigure}[t]{0.47\textwidth}
				\centering
				\subcaption{via left-curtain transport}
				\begin{tikzpicture}[x=5.15cm,y=0.68cm]
					\draw[->,line width=0.55pt]
					(-0.02,6.05)--(1.04,6.05)
					node[right,font=\scriptsize] {$m$};
					\node[left,font=\scriptsize]
					at (-0.02,6.05) {$\mu$};
					
					\draw[curtain blue]
					(0.178,6.05) rectangle (0.222,6.50);
					\draw[curtain orange]
					(0.478,6.05) rectangle (0.522,6.50);
					\draw[curtain green]
					(0.628,6.05) rectangle (0.672,6.82);
					
					\node[above,font=\scriptsize,text=blue!65!black]
					at (0.20,6.50) {$x_1$};
					\node[above,font=\scriptsize,text=orange!85!black]
					at (0.50,6.50) {$x_2$};
					\node[above,font=\scriptsize,text=green!45!black]
					at (0.65,6.82) {$x_3$};
					
					\draw[curtain green]
					(0,4.15) rectangle (0.10,4.43);
					\draw[curtain blue]
					(0.10,4.15) rectangle (0.30,4.43);
					\draw[curtain green]
					(0.30,4.15) rectangle (0.40,4.43);
					\draw[curtain orange]
					(0.40,4.15) rectangle (0.60,4.43);
					\draw[curtain green]
					(0.60,4.15) rectangle (1,4.43);
					
					\draw[->,line width=0.55pt]
					(-0.02,4.15)--(1.04,4.15)
					node[right,font=\scriptsize] {$x$};
					\node[left,font=\scriptsize]
					at (-0.02,4.15) {$\nu$};
					\node[above,font=\scriptsize,text=blue!65!black]
					at (0.20,4.34) {$D_1^L$};
					\node[above,font=\scriptsize,text=orange!85!black]
					at (0.50,4.34) {$D_2^L$};
					\node[above,font=\scriptsize,text=green!45!black]
					at (0.80,4.34) {$D_3^L$};
					
					\draw[->,line width=0.55pt]
					(-0.035,0.25)--(1.04,0.25)
					node[right,font=\scriptsize] {$x$};
					\draw[->,line width=0.55pt]
					(-0.035,0.25)--(-0.035,2.85);
					\node[anchor=west,font=\scriptsize]
					at (0.04,2.62) {$T_L(x)$};
					
					\draw[blue!65!black,line width=1.2pt,line cap=butt]
					(0.10,0.80)--(0.30,0.80);
					\draw[orange!85!black,line width=1.2pt,line cap=butt]
					(0.40,1.50)--(0.60,1.50);
					\draw[green!45!black,line width=1.2pt,line cap=butt]
					(0,2.20)--(0.10,2.20);
					\draw[green!45!black,line width=1.2pt,line cap=butt]
					(0.30,2.20)--(0.40,2.20);
					\draw[green!45!black,line width=1.2pt,line cap=butt]
					(0.60,2.20)--(1,2.20);
					
					\node[left,font=\scriptsize]
					at (-0.045,0.80) {$x_1$};
					\node[left,font=\scriptsize]
					at (-0.045,1.50) {$x_2$};
					\node[left,font=\scriptsize]
					at (-0.045,2.20) {$x_3$};
					\node[below,font=\scriptsize]
					at (0,0.25) {$c_j$};
					\node[below,font=\scriptsize]
					at (1,0.25) {$d_j$};
				\end{tikzpicture}
			\end{subfigure}\hfill
			\begin{subfigure}[t]{0.47\textwidth}
				\centering
				\subcaption{via right-curtain transport}
				\begin{tikzpicture}[x=5.15cm,y=0.68cm]
					\draw[->,line width=0.55pt]
					(-0.02,6.05)--(1.04,6.05)
					node[right,font=\scriptsize] {$m$};
					\node[left,font=\scriptsize]
					at (-0.02,6.05) {$\mu$};
					
					\draw[curtain blue]
					(0.178,6.05) rectangle (0.222,6.50);
					\draw[curtain orange]
					(0.478,6.05) rectangle (0.522,6.50);
					\draw[curtain green]
					(0.628,6.05) rectangle (0.672,6.82);
					
					\node[above,font=\scriptsize,text=blue!65!black]
					at (0.20,6.50) {$x_1$};
					\node[above,font=\scriptsize,text=orange!85!black]
					at (0.50,6.50) {$x_2$};
					\node[above,font=\scriptsize,text=green!45!black]
					at (0.65,6.82) {$x_3$};
					
					\draw[curtain blue]
					(0,4.15) rectangle (0.175,4.43);
					\draw[curtain orange]
					(0.175,4.15) rectangle (0.30,4.43);
					\draw[curtain green]
					(0.30,4.15) rectangle (0.90,4.43);
					\draw[curtain orange]
					(0.90,4.15) rectangle (0.975,4.43);
					\draw[curtain blue]
					(0.975,4.15) rectangle (1,4.43);
					
					\draw[->,line width=0.55pt]
					(-0.02,4.15)--(1.04,4.15)
					node[right,font=\scriptsize] {$x$};
					\node[left,font=\scriptsize]
					at (-0.02,4.15) {$\nu$};
					\node[above,font=\scriptsize,text=blue!65!black]
					at (0.08,4.34) {$D_1^R$};
					\node[above,font=\scriptsize,text=orange!85!black]
					at (0.238,4.34) {$D_2^R$};
					\node[above,font=\scriptsize,text=green!45!black]
					at (0.60,4.34) {$D_3^R$};
					
					\draw[->,line width=0.55pt]
					(-0.035,0.25)--(1.04,0.25)
					node[right,font=\scriptsize] {$x$};
					\draw[->,line width=0.55pt]
					(-0.035,0.25)--(-0.035,2.85);
					\node[anchor=west,font=\scriptsize]
					at (0.04,2.62) {$T_R(x)$};
					
					\draw[blue!65!black,line width=1.2pt,line cap=butt]
					(0,0.80)--(0.175,0.80);
					\draw[blue!65!black,line width=1.2pt,line cap=butt]
					(0.975,0.80)--(1,0.80);
					\draw[orange!85!black,line width=1.2pt,line cap=butt]
					(0.175,1.50)--(0.30,1.50);
					\draw[orange!85!black,line width=1.2pt,line cap=butt]
					(0.90,1.50)--(0.975,1.50);
					\draw[green!45!black,line width=1.2pt,line cap=butt]
					(0.30,2.20)--(0.90,2.20);
					
					\node[left,font=\scriptsize]
					at (-0.045,0.80) {$x_1$};
					\node[left,font=\scriptsize]
					at (-0.045,1.50) {$x_2$};
					\node[left,font=\scriptsize]
					at (-0.045,2.20) {$x_3$};
					\node[below,font=\scriptsize]
					at (0,0.25) {$c_j$};
					\node[below,font=\scriptsize]
					at (1,0.25) {$d_j$};
				\end{tikzpicture}
			\end{subfigure}
			\label{Fig: Three-point-distribution}
		\end{figure}
		\endgroup
	\end{example}
	
	The next example shows how the barcode transport works when the target distribution has a continuum of posterior means with a non-separating interval. 
	
	\begin{example}[\textbf{Normal distributions}]\label{Example: Normal}
		\normalfont
		Let $G=N(0,\sigma_G^2)$ and $F=N(0,\sigma_F^2)$, with $r:=\sigma_F/\sigma_G>1$, and denote their probability measures by $\mu,\nu$ and their densities by $g,f$. 
		The wider normal is a mean-preserving spread of the target, with $I_G(z)<I_F(z)$ for every finite $z$. 
		Thus the entire real line is one non-separating interval, and the target has a positive continuous density throughout it.
		
		To implement the barcode transport, we first construct the bars. 
		For these normal distributions, the set of points where $\mu$ density-dominates $\nu$ is
		\begin{equation*}
			C_1 = \left\{x:\ \frac{d\mu}{d(\mu+\nu)}(x)\geq\frac{1}{2}\right\} = \{x:g(x)\geq f(x)\}=[-c,c], 
			\quad \text{where} \quad
			c=\sigma_G\sqrt{\frac{2\log r}{1-r^{-2}}}\ . 
		\end{equation*}
		Let $\mu_1=\mu|_{C_1}$. 
		Its shadow on $\nu$ is $\nu_1=\nu|_{D_1}$ with $D_1=[-rc,rc]$. 
		The first pair of bars is thus $(\mu_1,\nu_1)$.
		
		After removing this pair, the remaining target measure $\tilde{\mu}_2=\mu-\mu_1=\mu|_{\{x:|x|>c\}}$ has positive mass only outside $[-c,c]$, whereas the remaining prior measure $\tilde{\nu}_2=\nu-\nu_1=\nu|_{\{x:|x|>rc\}}$ has positive mass only outside $[-rc,rc]$. 
		Hence, the next density-dominant region is
		\begin{align*}
			C_2 &= \left\{x:\ \frac{d\tilde{\mu}_2}{d(\tilde{\mu}_2+\tilde{\nu}_2)}(x)\geq\frac{1}{2}\right\} \\
			&= \{x:\ g(x)\1_{\{|x|>c\}}\geq f(x)\1_{\{|x|>rc\}}\} \\
			&= \{x:c<|x|\leq rc\} = [-rc,-c)\cup(c,rc] 
		\end{align*}
		Let $\mu_2=\tilde{\mu}_2|_{C_2}$. Its shadow on $\tilde{\nu}_2$ is 
		\begin{equation*}
			\nu_2 = S^{\tilde{\nu}_2}(\mu_2) = \nu|_{D_2}, \quad \text{where} \quad D_2 = \{x:rc<|x|\leq r^2c\} = [-r^2c,-rc)\cup(rc,r^2c]. 
		\end{equation*}
		The second pair is then $(\mu_2,\nu_2)$. 
		Figure~\ref{fig:normal_barcode_step0} illustrates the construction of the first two pairs of barcode bars for $\sigma_G=1$ and $r=2$.
		
		More generally, for $n\geq2$, let $\mu_n=\mu|_{C_n}$ and $\nu_n=\nu|_{D_n}$, where
		\begin{equation*}
			C_n = \{x:\ r^{n-2}c<|x|\leq r^{n-1}c\}
			\qquad \text{and} \qquad 
			D_n = rC_n=\{x:\ r^{n-1}c<|x|\leq r^nc\}.
		\end{equation*}
		Then $(\mu_n,\nu_n)$ forms the $n$-th pair of bars. 
		Both have equal mass and zero barycenter, and $\nu_n$ is a symmetric outward dilation of $\mu_n$. 
		Moreover, $D_n$ takes the innermost remaining prior states with that mass, so its restriction is the shadow of $\mu|_{C_n}$. 
		Since $\cup_{n=1}^{\infty}{C_n}=\cup_{n=1}^{\infty}{D_n}=\R$, we have $\mu=\sum_{n=1}^{\infty}{\mu_n}$ and $\nu=\sum_{n=1}^{\infty}{\nu_n}$ so the countable collections of bars $\{(\mu_n,\nu_n)\}_n$ indeed exhaust the target and prior measures $\mu$ and $\nu$.  
		
		\begin{figure}[htbp]
			\centering
			\caption{The first two pairs of barcode bars and their implementations}
			\begin{subfigure}[t]{0.45\textwidth}
				\centering
				\subcaption{The first pair of bars.}
				\begin{tikzpicture}
					\begin{axis}[
						width=0.80\linewidth,
						height=0.2\textheight,
						scale only axis,
						axis lines=left,
						axis line style={line width=0.8pt,-{Stealth[length=2mm]}},
						xmin=-3.15,xmax=3.15,
						ymin=-0.005,ymax=0.43,
						xtick={-2.719112,-1.359556,0,1.359556,2.719112},
						xticklabels={$-2c$,$-c$,$0$,$c$,$2c$},
						ytick={},
						label style={font=\small},
						tick label style={font=\footnotesize},
						legend style={at={(0.03,0.98)},anchor=north west,
							legend cell align=left,fill=none,draw=none},
						clip=false,
						samples=180
						]
						\addplot[
						draw=none,fill=black!25,fill opacity=0.35,
						domain=-2.719112:2.719112,forget plot
						] {1/(2*sqrt(2*pi))*exp(-x^2/8)} \closedcycle;
						\addplot[
						draw=none,fill=red!35,fill opacity=0.35,
						domain=-1.359556:1.359556,forget plot
						] {1/sqrt(2*pi)*exp(-x^2/2)} \closedcycle;
						\addplot[black,thick,domain=-3.15:3.15]
						{1/(2*sqrt(2*pi))*exp(-x^2/8)};
						\addlegendentry{\small $f$}
						\addplot[red,thick,domain=-3.15:3.15]
						{1/sqrt(2*pi)*exp(-x^2/2)};
						\addlegendentry{\small \textcolor{red}{$g$}}
						\node[font=\small,text=red!70!black] at (axis cs:0,0.075) {$C_1$};
						\node[font=\small] at (axis cs:2.15,0.075) {$D_1$};
					\end{axis}
				\end{tikzpicture}
			\end{subfigure}\hfill
			\begin{subfigure}[t]{0.45\textwidth}
				\centering
				\subcaption{The second pair of bars.}
				\begin{tikzpicture}
					\begin{axis}[
						width=0.80\linewidth,
						height=0.20\textheight,
						scale only axis,
						axis lines=left,
						axis line style={line width=0.8pt,-{Stealth[length=2mm]}},
						xmin=-6.05,xmax=6.05,
						ymin=-0.003,ymax=0.19,
						xtick={-5.438224,-2.719112,-1.359556,0,1.359556,2.719112,5.438224},
						xticklabels={$-4c$,$-2c$,$-c$,$0$,$c$,$2c$,$4c$},
						ytick={},
						label style={font=\small},
						tick label style={font=\scriptsize},
						legend style={at={(0.03,0.98)},anchor=north west,
							legend cell align=left,fill=none,draw=none},
						clip=false,
						samples=120
						]
						\addplot[
						draw=none,fill=black!25,fill opacity=0.35,
						domain=-5.438224:-2.719112,forget plot
						] {1/(2*sqrt(2*pi))*exp(-x^2/8)} \closedcycle;
						\addplot[
						draw=none,fill=black!25,fill opacity=0.35,
						domain=2.719112:5.438224,forget plot
						] {1/(2*sqrt(2*pi))*exp(-x^2/8)} \closedcycle;
						\addplot[
						draw=none,fill=red!35,fill opacity=0.35,
						domain=-2.719112:-1.359556,forget plot
						] {1/sqrt(2*pi)*exp(-x^2/2)} \closedcycle;
						\addplot[
						draw=none,fill=red!35,fill opacity=0.35,
						domain=1.359556:2.719112,forget plot
						] {1/sqrt(2*pi)*exp(-x^2/2)} \closedcycle;
						
						\addplot[black,thick,domain=-6.05:-2.719112]
						{1/(2*sqrt(2*pi))*exp(-x^2/8)};
						\addlegendentry{\small $f$}
						\addplot[black,thick,domain=2.719112:6.05,forget plot]
						{1/(2*sqrt(2*pi))*exp(-x^2/8)};
						\addplot[red,thick,domain=-6.05:-1.359556]
						{1/sqrt(2*pi)*exp(-x^2/2)};
						\addlegendentry{\small \textcolor{red}{$g$}}
						\addplot[red,thick,domain=1.359556:6.05,forget plot]
						{1/sqrt(2*pi)*exp(-x^2/2)};
						\node[font=\small,text=red!70!black] at (axis cs:1.9,0.025) {$C_2$};
						\node[font=\small] at (axis cs:3.5,0.025) {$D_2$};
					\end{axis}
				\end{tikzpicture}
			\end{subfigure}
		\end{figure}
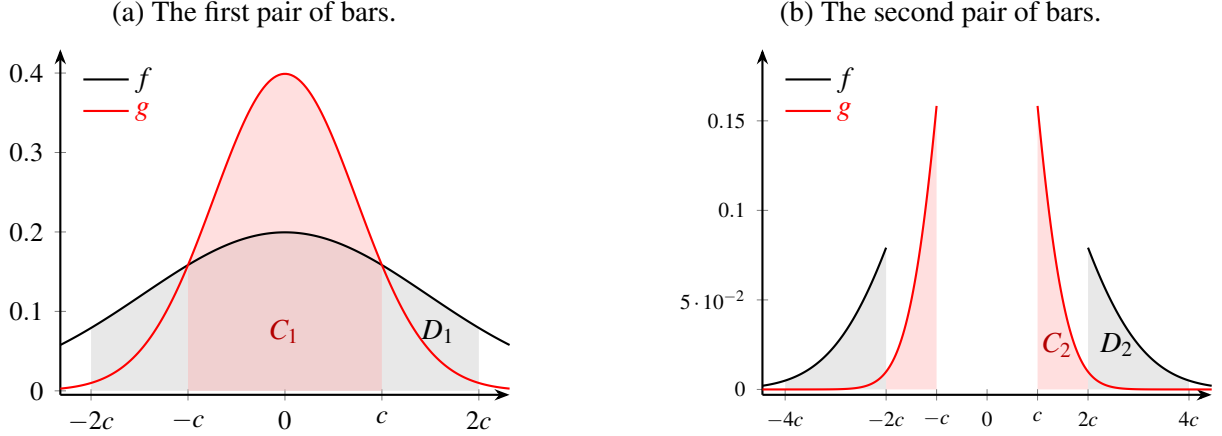
		
		The second barcode step constructs $T_n$ from the left-curtain transport within each pair $(\mu_n,\nu_n)$. For a target mean level $s$ between the endpoints of $C_n$, the target's lower contour has shadow
		\begin{equation*}
			S^{\nu_n}\!\left(\mu_n|_{(-\infty,s]}\right)
			=\nu_n|_{[a_n(s),b_n(s)]},
		\end{equation*}
		where the boundaries are determined by 
		\begin{align}
			\int_{D_n\cap[a_n(s),b_n(s)]}f(x)\,dx &= \int_{C_n\cap(-\infty,s]}g(m)\,dm, \label{Eq: Normal_mass_equality} \\
			\int_{D_n\cap[a_n(s),b_n(s)]}x f(x)\,dx &= \int_{C_n\cap(-\infty,s]}m g(m)\,dm. \label{Eq: Normal_mean_equality}
		\end{align}
		In the first pair $(\mu_1,\nu_1)$, the lower contour set $\{x\in D_1:\ T_1(x)\leq s\}$ is an interval $[a_1(s),b_1(s)]$, which expands from the singleton $\{-c\}$ to $D_1=[-rc,rc]$ as $s$ rises in $C_1=[-c,c]$. 
		More generally, in the $n$-th pair with $n\geq2$, the lower contour set $\{x\in D_n:\ T_n(x)\leq s\}$ is a union of two intervals
		\begin{equation*}
			[a_n(s),-r^{n-1}c]\ \cup\ [r^{n-1}c,b_n(s)],
		\end{equation*}
		which expands from the inner edges of $D_n$ to its outer edges as $s$ increases on $C_n$. 
		It stays fixed as $s$ crosses the gap between the two pieces of $C_n$. Thus $a_n$ decreases and $b_n$ increases on each part of the target support. 
		Figure~\ref{Fig: Normal_map} illustrates this for the first two pairs of barcode bars. 
		
		\begin{figure}[htbp]
			\centering
			\caption{A barcode implementation using left-curtain transport within the first two bars} 
			\begin{tikzpicture}
				\begin{axis}[
					width=0.84\textwidth,
					height=0.24\textheight,
					scale only axis,
					axis lines=middle,
					axis line style={line width=0.8pt,-{Stealth[length=2mm]}},
					xmin=-4.35,xmax=4.35,
					ymin=-2.35,ymax=2.35,
					xtick={-4,-2,0,2,4},
					xticklabels={$-r^2c$,$-rc$,$0$,$rc$,$r^2c$},
					ytick={-2,-1,0,1,2},
					yticklabels={$-rc$,$-c$,$0$,$c$,$rc$},
					xlabel={$x$},
					xlabel style={at={(axis description cs:1.015,0.455)},anchor=south,rotate=0},
					ylabel={$T(x)$},
					ylabel style={at={(axis description cs:0.5,1.0)},anchor=south,rotate=0},
					label style={font=\footnotesize},
					tick label style={font=\footnotesize},
					xticklabel style={anchor=south,yshift=2pt}, 
					clip=false
					]
					\draw[black!25,densely dotted,line width=0.6pt]
					(axis cs:-2,-2.30)--(axis cs:-2,2.30);
					\draw[black!25,densely dotted,line width=0.6pt]
					(axis cs:2,-2.30)--(axis cs:2,2.30);
					
					\addplot[blue!70!black,very thick,smooth] coordinates {
						(-2.000000,1.000) (-1.932250,0.875)
						(-1.854931,0.750) (-1.769953,0.625)
						(-1.679742,0.500) (-1.587002,0.375)
						(-1.494502,0.250) (-1.404918,0.125)
						(-1.320701,0.000) (-1.243987,-0.125)
						(-1.176518,-0.250) (-1.119572,-0.375)
						(-1.073912,-0.500) (-1.039742,-0.625)
						(-1.016703,-0.750) (-1.003905,-0.875)
						(-1.000000,-1.000)
					};
					\addplot[blue!70!black,very thick,smooth] coordinates {
						(-1.000000,-1.000) (-0.867923,-0.875)
						(-0.722353,-0.750) (-0.564443,-0.625)
						(-0.395485,-0.500) (-0.216832,-0.375)
						(-0.029846,-0.250) (0.164135,-0.125)
						(0.363808,0.000) (0.567900,0.125)
						(0.775138,0.250) (0.984214,0.375)
						(1.193717,0.500) (1.402070,0.625)
						(1.607432,0.750) (1.807615,0.875)
						(2.000000,1.000)
					};
					
					\addplot[orange!85!black,very thick,smooth] coordinates {
						(-4.000000,2.000) (-3.926273,1.875)
						(-3.829400,1.750) (-3.710403,1.625)
						(-3.572604,1.500) (-3.420570,1.375)
						(-3.259044,1.250) (-3.092300,1.125)
						(-2.923919,1.000)
					};
					\addplot[orange!85!black,very thick,smooth] coordinates {
						(-2.923919,-1.000) (-2.590418,-1.125)
						(-2.386978,-1.250) (-2.251129,-1.375)
						(-2.157744,-1.500) (-2.093299,-1.625)
						(-2.049218,-1.750) (-2.019553,-1.875)
						(-2.000000,-2.000)
					};
					\addplot[orange!85!black,very thick,smooth] coordinates {
						(2.000000,-2.000) (2.000379,-1.875)
						(2.002087,-1.750) (2.006437,-1.625)
						(2.015604,-1.500) (2.033036,-1.375)
						(2.064093,-1.250) (2.117169,-1.125)
						(2.206332,-1.000)
					};
					\addplot[orange!85!black,very thick,smooth] coordinates {
						(2.206332,1.000) (2.440328,1.125)
						(2.674536,1.250) (2.908133,1.375)
						(3.139953,1.500) (3.368264,1.625)
						(3.590483,1.750) (3.802816,1.875)
						(4.000000,2.000)
					};
					
					\addplot[red!65!black,dashed,line width=0.8pt] coordinates {
						(-1.679742,0.5) (1.193717,0.5)
					};
					\addplot[red!65!black,only marks,mark=*,mark size=1.5pt]
					coordinates {(-1.679742,0.5) (1.193717,0.5)};
					\node[font=\footnotesize,text=red!65!black,anchor=south east]
					at (axis cs:-0.60,0.54) {$s_1$};
					
					\addplot[red!65!black,dashed,line width=0.8pt] coordinates {
						(-3.572604,1.5) (3.139953,1.5)
					};
					\addplot[red!65!black,only marks,mark=*,mark size=1.5pt]
					coordinates {(-3.572604,1.5) (3.139953,1.5)};
					\node[font=\footnotesize,text=red!65!black,anchor=south east]
					at (axis cs:-0.60,1.54) {$s_2$};
					
					\draw[red!65!black,densely dotted,line width=0.6pt]
					(axis cs:-1.679742,0.5)--(axis cs:-1.679742,0)
					(axis cs: 1.193717,0.5)--(axis cs: 1.193717,0)
					(axis cs:-3.572604,1.5)--(axis cs:-3.572604,0)
					(axis cs: 3.139953,1.5)--(axis cs: 3.139953,0);
					
					\addplot[red!65!black,only marks,mark=*,mark size=1.2pt]
					coordinates {
						(-1.679742,0) (1.193717,0)
						(-3.572604,0) (3.139953,0)
					};
					
					\node[font=\footnotesize,text=red!65!black,
					anchor=north,yshift=-3pt]
					at (axis cs:-1.679742,0) {$a_1(s_1)$};
					
					\node[font=\footnotesize,text=red!65!black,
					anchor=north,yshift=-3pt]
					at (axis cs:1.193717,0) {$b_1(s_1)$};
					
					\node[font=\footnotesize,text=red!65!black,
					anchor=north,yshift=-3pt]
					at (axis cs:-3.572604,0) {$a_2(s_2)$};
					
					\node[font=\footnotesize,text=red!65!black,
					anchor=north,yshift=-3pt]
					at (axis cs:3.139953,0) {$b_2(s_2)$};
				\end{axis}
			\end{tikzpicture}
			\label{Fig: Normal_map}
			\vspace*{0.75ex}
			\begin{minipage}{1.0\textwidth}
				\footnotesize 
				Note: Blue branches show $T_1$ on $D_1=[-rc,rc]$, and orange branches show $T_2$ on $D_2=[-r^2c,rc)\cup(rc,r^2c]$. 
			\end{minipage}
		\end{figure}
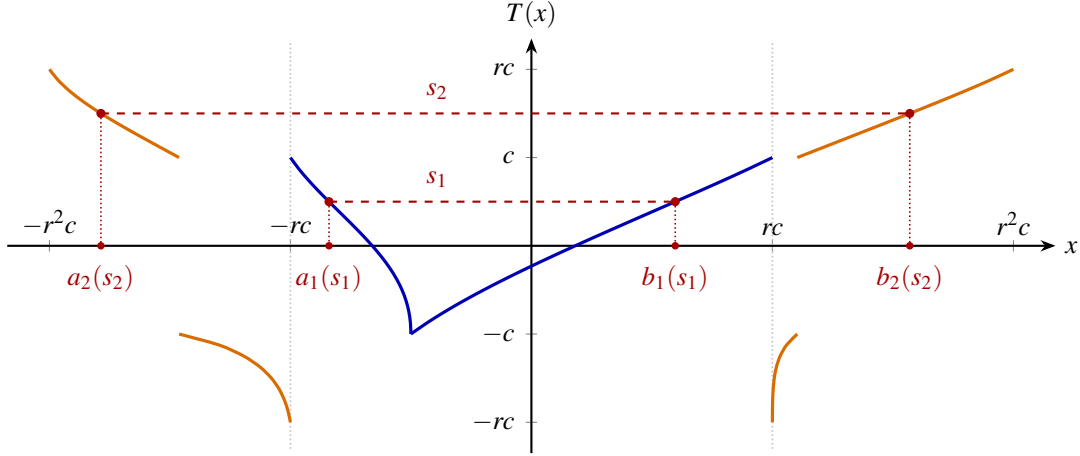
		
		Each state is assigned the posterior mean at which it enters the expanding contour:
		\begin{equation*}
			T_n(x)=\inf\{s\in C_n:x\in[a_n(s),b_n(s)]\} \qquad \text{for all} \quad x\in D_n.
		\end{equation*}
		The two moment constraints \eqref{Eq: Normal_mass_equality} and \eqref{Eq: Normal_mean_equality} ensure that this mapping pushes $\nu_n$ forward to $\mu_n$ and satisfies $\E[X|T_n(X)]=T_n(X)$. 
		Pasting the maps over $D_n$ yields $T(X)\sim G$ and $\E[X|T(X)]=T(X)$. Therefore, $T$ implements $G$ from $F$.  
		\qed 
	\end{example}
	
	These examples highlight two distinct uses of the construction. A $K$-pooling target requires one bar and at most $K$ posterior-mean cells on each non-separating interval. The normal target requires countably many bars and a continuum of posterior means within each bar. The normal contours are nested within individual bars, but for $n\geq2$ they are unions of two state intervals; pasting the bars does not ensure interval contour sets on the whole non-separating component. The next section characterizes when implementation can instead be organized through lower or upper contour sets that are ordinary nested intervals on each component.
	
	\section{Contour nested-interval implementation}\label{Sec: Contour_nested_interval}
	The barcode construction guarantees deterministic implementation, but its contour sets need not be intervals. In this section we characterize when a contour nested-interval implementation is feasible and provide conditions on the prior and target distributions that help determine which contour direction to use.
	Throughout, we fix an atomless prior $F$ with finite first moment and a target $G\in\MPC(F)$, and write $\mathbb{NS}_G=\cup_{j\in J}(c_j,d_j)$ for its non-separating intervals.
	
	\begin{definition}\label{Definition: Contour_nested_interval}
		Let $T$ implement $G$ from $F$. On a non-separating interval $(c_j,d_j)$, $T$ is
		\begin{itemize}
			\item a \emph{\textbf{lower contour nested-interval implementation}} if, for every $s\in(c_j,d_j)$, the lower contour set $\{x\in(c_j,d_j):T(x)\leq s\}$ is an interval;
			\item an \emph{\textbf{upper contour nested-interval implementation}} if, for every $s\in(c_j,d_j)$, the upper contour set $\{x\in(c_j,d_j):T(x)\geq s\}$ is an interval.
		\end{itemize}
		We call $T$ a \emph{\textbf{contour nested-interval experiment}} if it has one of these two structures on every non-separating interval. The direction may differ across intervals.\footnote{Throughout this section, contour sets and interval restrictions are identified up to prior-null sets. Since $F$ is atomless, interval endpoint conventions are immaterial. Empty sets and the entire component are also counted as intervals.}
	\end{definition}
	
	Fix a non-separating interval $(c_j,d_j)$. As in Section~\ref{Subsec: Barcode_Transport}, let $\mu$ and $\nu$ denote the probability measures of $G|_{(c_j,d_j)}$ and $F|_{(c_j,d_j)}$, respectively. By Proposition~\ref{Prop: Decomposition}, $\mu\leq_{cx}\nu$, and the prior and target have the same mass
	\begin{equation*}
		p_j:=F(d_j)-F(c_j)=G(d_j)-G(c_j)>0
	\end{equation*}
	on this interval. We first introduce a useful auxiliary result.
	
	\begin{lemma}\label{Lemma: Moment_matched_interval}
		For every submeasure $0\leq\alpha\leq\mu$, there is an interval $I\subseteq[c_j,d_j]$ such that $\nu|_I$ has the same mass and first moment as $\alpha$.
		This interval restriction is unique up to $\nu$-null sets. Moreover, among all submeasures of $\nu$ with these two moments, $\nu|_I$ is minimal in convex order.
	\end{lemma}
	\begin{proof}[Proof of Lemma~\ref{Lemma: Moment_matched_interval}]
		See Appendix~\ref{App. Proof_Lemma_Interval}.
	\end{proof}
	
	For lower contour nested-interval implementation, define $\alpha_s:=\mu|_{(c_j,s]}$ for $s\in(c_j,d_j)$. Lemma~\ref{Lemma: Moment_matched_interval} identifies a unique prior interval restriction $\beta_s:=\nu|_{I_j(s)}$ with the same mass and first moment. For a nonempty interval $I_j(s)=[a_j(s),b_j(s)]$, its boundaries satisfy
	\begin{align}
		F(b_j(s))-F(a_j(s))&=G(s)-G(c_j), \label{Eq: Lower_contour_mass}\\
		\int_{a_j(s)}^{b_j(s)}x\,dF(x)&=\int_{(c_j,s]}m\,dG(m). \label{Eq: Lower_contour_barycenter}
	\end{align}
	Proposition~\ref{Prop: Nested_Contours} gives necessary and sufficient conditions for these intervals to be the lower contour sets of an implementing signal.
	
	\begin{proposition}\label{Prop: Nested_Contours}
		Let $(c_j,d_j)$ be a non-separating interval of $G$ and, for each $s\in(c_j,d_j)$, let $I_j(s)$ be the interval determined by \eqref{Eq: Lower_contour_mass} and \eqref{Eq: Lower_contour_barycenter}.
		Then the following statements are equivalent: 
		\begin{enumerate}
			\item[(a)] There exists a lower contour nested-interval implementation on $(c_j,d_j)$.
			\item[(b)] The intervals $I_j(s)$ expand in $s$: $I_j(s')\subseteq I_j(s'')$ for $s'<s''$, up to $\nu$-null sets.
			\item[(c)] $\alpha_s\leq_{cx}\beta_s$ for every $s\in(c_j,d_j)$.
			\item[(d)] $S^\nu(\alpha_s)=\beta_s$ for every $s\in(c_j,d_j)$.
		\end{enumerate}
	\end{proposition}
	\begin{proof}[Proof of Proposition~\ref{Prop: Nested_Contours}]
		See Appendix~\ref{App. Proof_Prop_Nested_Contours}.
	\end{proof}
	
	Proposition~\ref{Prop: Nested_Contours} connects the feasibility of lower contour nested-interval implementation to the geometry of its contour sets. For each target lower contour $\alpha_s$, Lemma~\ref{Lemma: Moment_matched_interval} identifies a unique interval restriction $\beta_s$ with the required mass and first moment. Any lower contour nested-interval implementation must therefore use these intervals. Since $\{T(X)\in(c_j,s]\}$ expands with $s$, $I_j(s)$ must expand as well. The equivalence between (a) and (b) shows that this geometric consistency is also sufficient to assemble the intervals into an implementing signal.
	
	Part (c) gives an information-feasibility interpretation of this requirement. Matching mass and first moment ensures that an interval reproduces the probability and average posterior mean of the target contour set. Generating the entire distribution of posterior means within that contour set additionally requires its state distribution to be sufficiently dispersed, as expressed by $\alpha_s\leq_{cx}\beta_s$. The equivalence shows that feasibility within every interval also ensures compatibility across thresholds. Part (d) explains why: by Lemma~\ref{Lemma: Moment_matched_interval}, $\beta_s$ is the least dispersed submeasure of $\nu$ with the prescribed mass and first moment, so whenever it dominates $\alpha_s$ in convex order, it must be its shadow. Shadows of increasing target contour sets accumulate without undoing earlier allocations, yielding the expansion in (b). Moreover, matching every lower target contour to its shadow characterizes the left-curtain coupling, connecting the geometric criterion to a specific transport construction. 
	
	Consequently, whenever a lower contour nested-interval implementation is feasible, it is unique up to $\nu$-null sets and is obtained by reading the left-curtain transport backward. With nested interval representatives, every such implementation agrees almost surely with
	\begin{equation}\label{Eq: Lower_Contour_nested_interval}
		T_L(x)=\inf\{s\in(c_j,d_j):x\in I_j(s)\} 
		\quad \text{for $\nu$-a.e. } x\in(c_j,d_j).
	\end{equation}
	Thus, each state is assigned the first target mean level at which it enters the expanding family of intervals. On nonempty contour sets, the endpoints can be chosen with $a_j$ non-increasing and $b_j$ non-decreasing. 
	
	For upper contour nested-interval implementation, instead define $\alpha_s:=\mu|_{[s,d_j)}$. Again let $\beta_s:=\nu|_{I_j(s)}$ with $I_j(s)=[a_j(s),b_j(s)]$ be the unique interval restriction with the same mass and first moment. For this to hold, the interval boundaries $a_j(s)$ and $b_j(s)$ must satisfy 
	\begin{align}
		F(b_j(s))-F(a_j(s))&=G(d_j)-G(s^-), \label{Eq: Upper_contour_mass}\\
		\int_{a_j(s)}^{b_j(s)}x\,dF(x)&=\int_{[s,d_j)}m\,dG(m). \label{Eq: Upper_contour_barycenter}
	\end{align}
	The left limit $G(s^-)$ includes any target atom at $s$ in the upper contour set. Reflecting the state and posterior mean gives the upper contour counterpart to Proposition~\ref{Prop: Nested_Contours}.\footnote{We omit the proof of Proposition~\ref{Prop: Nested_Contours_upper} because it is similar to the proof of Proposition~\ref{Prop: Nested_Contours}.}
	
	\begin{propositionprime}\label{Prop: Nested_Contours_upper}
		Let $(c_j,d_j)$ be a non-separating interval of $G$ and, for each $s\in(c_j,d_j)$, let $I_j(s)$ be the interval determined by \eqref{Eq: Upper_contour_mass} and \eqref{Eq: Upper_contour_barycenter}.
		Then the following statements are equivalent:
		\begin{enumerate}
			\item[(a)] There exists an upper contour nested-interval implementation on $(c_j,d_j)$.
			\item[(b)] The intervals $I_j(s)$ shrink in $s$: $I_j(s'')\subseteq I_j(s')$ for $s'<s''$, up to $\nu$-null sets.
			\item[(c)] $\alpha_s\leq_{cx}\beta_s$ for every $s\in(c_j,d_j)$.
			\item[(d)] $S^\nu(\alpha_s)=\beta_s$ for every $s\in(c_j,d_j)$.
		\end{enumerate}
	\end{propositionprime}
	
	When these conditions hold, reading the right-curtain transport backward gives the unique upper contour nested-interval implementation, up to $\nu$-null sets. With nested interval representatives, it is given almost surely by
	\begin{equation}\label{Eq: Upper_Contour_nested_interval}
		T_R(x)=\sup\{s\in(c_j,d_j):x\in I_j(s)\} 
		\quad \text{for $\nu$-a.e. } x\in(c_j,d_j).
	\end{equation}
	On nonempty contour sets, $a_j$ can be chosen non-decreasing and $b_j$ non-increasing. As before, assign a fixed value in the component on the remaining null set.
	
	These characterizations give a simple procedure to check feasibility of contour nested-interval implementations within a non-separating interval of the target distribution: first solve the two moment-preserving equations, and then check whether the intervals expand for lower contour sets or shrink for upper contour sets. 
	By Proposition~\ref{Prop: Decomposition}, a global contour nested-interval experiment exists exactly when at least one of the two characterizations holds on every non-separating interval. The local maps can then be combined with full revelation on $\mathbb S_G$. We next provide further conditions that help determine which contour direction can work. 
	
	\begin{proposition}\label{Prop: Contour_support_gap}
		Suppose $T$ implements $G$ from $F$ and let $(c_j,d_j)$ be a non-separating interval of $G$. 
		Define
		\begin{equation*}
			\ell_j:=\inf\supp\left(G|_{(c_j,d_j)}\right)
			\qquad\text{and}\qquad
			r_j:=\sup\supp\left(G|_{(c_j,d_j)}\right).
		\end{equation*}
		Then the following properties holds:
		\begin{enumerate}
			\item If $T$ is a lower contour nested-interval implementation on $(c_j,d_j)$, then $\ell_j>c_j$.
			\item If $T$ is an upper contour nested-interval implementation on $(c_j,d_j)$, then $r_j<d_j$.
		\end{enumerate}
	\end{proposition}
	\begin{proof}[Proof of Proposition~\ref{Prop: Contour_support_gap}]
		See Appendix~\ref{App. Proof_Contour_support_gap}.
	\end{proof}
	
	By Proposition~\ref{Prop: Contour_support_gap}, a target distribution whose support reaches the lower bound of a non-separating interval rules out the feasibility of lower contour nested-interval implementation. 
	Similarly, a target distribution whose support reaches the upper bound of a non-separating interval rules out the feasibility of upper contour nested-interval implementation. 
	This yields the following corollary. 
	
	\begin{corollary}\label{Corollary: Impossible_Contour_interval}
		Let $T$ implement $G$ from $F$ and let $(c_j,d_j)$ be a non-separating interval. If $\ell_j=c_j$ and $r_j=d_j$, then $T$ is neither a lower nor an upper contour nested-interval implementation on this component.
	\end{corollary}
	
	Example~\ref{Example: Normal} in Section~\ref{Subsec: Examples_Construction} illustrates this obstruction. 
	There, both the prior and target are normal distributions that have full support on their unique non-separating interval $(-\infty,\infty)$, so neither contour direction is feasible on the whole component. The barcode implementation in Figure~\ref{Fig: Normal_map} remains deterministic, but nesting within individual bars does not give interval contour sets on the whole real line.
	
	Finally, we give sufficient conditions directly in terms of the prior and target densities and cdfs. For Proposition~\ref{Prop: Contour_density_cdf} only, fix a bounded non-separating interval $(c_j,d_j)$ of $G$. Suppose that $F$ has a continuous and strictly positive density $f$ on this interval and that the local target has connected support $[\ell_j,r_j]\subseteq[c_j,d_j]$, with $\ell_j<r_j$. Suppose also that its probability measure can be written as
	\begin{equation*}
		\mu=\frac{1}{p_j}\left(
		g(m)\1_{(\ell_j,r_j)}(m)\,dm
		+\sum_{k=1}^{K}q_k\delta_{x_k}
		\right),
	\end{equation*}
	where $g$ is continuous and non-negative on $(\ell_j,r_j)$, $K$ is finite and may be zero, and $q_k>0$ at distinct points $x_k\in[\ell_j,r_j]\cap(c_j,d_j)$. Here $p_j$ is the common mass of the prior and target on the component. Thus $f$, $g$, and the atom masses $q_k$ refer to the original distributions, before conditioning on $(c_j,d_j)$. These regularity conditions are not required for the preceding characterizations. In the conditions below, a cutoff at a support endpoint makes one of the two inequalities vacuous.
	
	\begin{proposition}\label{Prop: Contour_density_cdf}
		Under the above assumptions, the following properties hold:
		\begin{enumerate}
			\item If there exists $\tau_j\in[\ell_j,r_j]$ such that
			\begin{equation}\label{Eq: Contour_lower_test}
				\begin{cases}
					f(s)\leq g(s),&\ell_j<s<\tau_j,\\
					F(s)\leq G(s^-),&\tau_j<s<r_j,
				\end{cases}
			\end{equation}
			then the intervals determined by \eqref{Eq: Lower_contour_mass} and \eqref{Eq: Lower_contour_barycenter} expand with $s$. The backward left-curtain map $T_L$ gives the unique lower contour nested-interval implementation, up to $\nu$-null sets.
			\item If there exists $\tau_j\in[\ell_j,r_j]$ such that
			\begin{equation}\label{Eq: Contour_upper_test}
				\begin{cases}
					F(s)\geq G(s),&\ell_j<s<\tau_j,\\
					f(s)\leq g(s),&\tau_j<s<r_j,
				\end{cases}
			\end{equation}
			then the intervals determined by \eqref{Eq: Upper_contour_mass} and \eqref{Eq: Upper_contour_barycenter} shrink with $s$. The backward right-curtain map $T_R$ gives the unique upper contour nested-interval implementation, up to $\nu$-null sets.
		\end{enumerate}
	\end{proposition}
	\begin{proof}[Proof of Proposition~\ref{Prop: Contour_density_cdf}]
		See Appendix~\ref{App. Proof_Contour_density_cdf}.
	\end{proof}
	
	Proposition~\ref{Prop: Contour_support_gap} rules out a contour direction, while Proposition~\ref{Prop: Contour_density_cdf} gives sufficient conditions for a direction to work. When neither sufficient test applies, the expansion and shrinkage criteria in Propositions~\ref{Prop: Nested_Contours} and~\ref{Prop: Nested_Contours_upper} remain the general test. Section~\ref{Sec: Applications} uses these results to select the contour direction and then solves the boundary equations to describe the implementing maps.

	\section{Applications}\label{Sec: Applications}
	In this section we study deterministic implementation in two ``non-standard'' information design problems: buyer-optimal learning in monopolistic pricing \citep{Roesler2017buyer}, and robust persuasion with ambiguous receiver preferences \citep{SapiroGheiler2024}. 
	In both problems, the objectives are not linear functions of posterior-mean distributions. 
	Consequently, their optimal solutions need not lie in the set of extreme points of $\MPC(F)$, which coincides with the set of all bi-pooling distributions \citep{Kleiner2021extreme,Arieli2023optimal}. 
	In fact, for both problems, the optimal posterior-mean distributions have a continuum of support points within a non-separating component and are therefore not bi-pooling. 
	Nevertheless, these optimal distributions can be implemented by contour nested-intervals as described in Section~\ref{Sec: Contour_nested_interval}. 
	
	\subsection{Buyer-optimal learning}\label{Subsec: RoeslerSzentes2017} 
	In \citet{Roesler2017buyer}, a monopolistic seller with zero production cost sells one object to a buyer whose valuation $X$ has cdf $F$ on $[0,1]$. Before knowing the price, the buyer privately observes a signal about $X$. Let $M$ denote her posterior mean and $G$ its cdf. The seller observes neither $X$ nor the signal realization, but knows $G$, and posts a take-it-or-leave-it price $q$. Because the buyer purchases if and only if $M\geq q$, demand is $1-G(q^-)$ and the seller's profit is $q[1-G(q^-)]$. Bayes plausibility requires $G\in\MPC(F)$. Roesler and Szentes study the information structure that maximizes the buyer's ex-ante surplus subject to Bayes plausibility and the seller’s pricing response:
	\begin{equation}\label{Eq: Problem_Roesler_Szentes}
		\max_{G\in\MPC(F)}{\int^{1}_{q}(s-q)\,dG(s)} \quad \text{subject to} \quad q\in\arg\max_{r\in[0,1]}r[1-G(r^-)].
	\end{equation}
	They show that the buyer-optimal solution to \eqref{Eq: Problem_Roesler_Szentes} takes the truncated equal-revenue form 
	\begin{equation}\label{Eq: RS_cdf}
		G(s)=
		\begin{cases}
			0, & s<p,\\
			1-\dfrac{p}{s}, & p\leq s<B,\\
			1, & s\geq B.
		\end{cases}
	\end{equation}
	It has density $p/s^2$ on $[p,B)$ and an atom of mass $p/B$ at $B$. The demand is $p/q$ for every $q\in[p,B]$, and the seller earns the same revenue $p$ at every price in this interval. The buyer-optimal outcome selects the lowest best response $p$, at which trade occurs with probability one. 
	
	With a uniform prior $F$ on $[0,1]$, \cite{Roesler2017buyer} show that the parameters for the buyer-optimal distribution are given by
	\begin{equation*}
		p=L(1-L) \quad \text{and} \quad B=p\exp\left(\frac{1}{2p}-1\right), 
	\end{equation*}
	where $L\in(0,1)$ is the unique interior solution to equation
	\begin{equation*}
		(1-L)\log\frac{1}{1-L}=\frac{L}{2}\ . 
	\end{equation*}
	Numerically, $L\approx0.72$, $p\approx0.20$, and $B\approx0.87$. 
	
	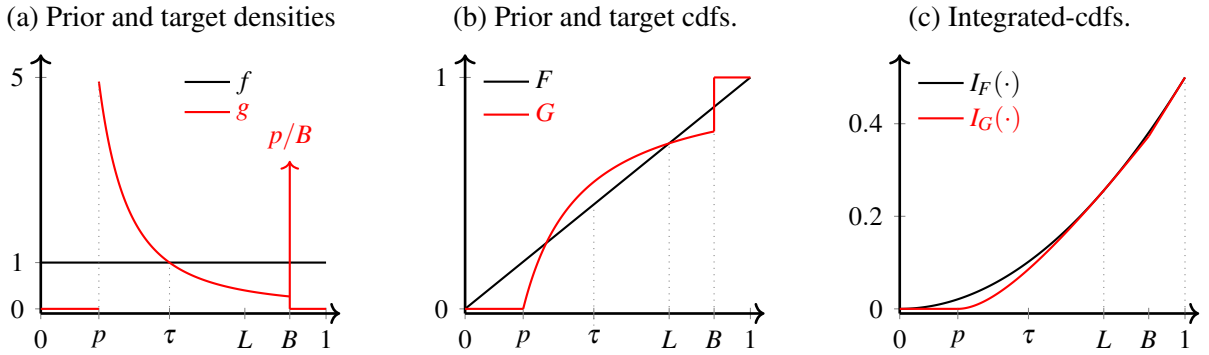
\begin{figure}[htbp]
		\centering
		\caption{The Roesler--Szentes buyer-optimal distribution under the uniform prior $F$}
		
		\pgfmathsetmacro{\RSparamP}{0.2036321887945369}
		\pgfmathsetmacro{\RSparamL}{0.7153318629591614}
		\pgfmathsetmacro{\RSparamB}{0.8728143298793253}
		\pgfmathsetmacro{\RSparamTau}{0.45125623407831} 
		
		\pgfplotsset{appdistaxis/.style={
				width=0.75\linewidth,
				height=0.15\textheight,
				scale only axis,
				axis lines=left,
				axis line style={line width=1pt,->},
				xmin=0,xmax=1.05,
				xtick={0,1},
				tick label style={font=\footnotesize},
				axis background/.style={fill=none},
				legend style={
					at={(0.03,0.98)},anchor=north west,
					font=\footnotesize,
					legend cell align=left,
					fill=none,draw=none
				},
				clip=false,
				samples=160
		}}
		
		\begin{subfigure}[t]{0.32\textwidth}
			\centering
			\subcaption{Prior and target densities}
			\begin{tikzpicture}
				\begin{axis}[appdistaxis,
					ymin=-0.10,ymax=5.50,
					ytick={0,1,5},
					extra x ticks={\RSparamP,\RSparamTau,\RSparamL,\RSparamB},
					extra x tick labels={$p$,$\tau$,$L$,$B$},
					legend style={
						at={(0.75,0.98)},anchor=north east
					}
					]
					\addplot[black,thick,domain=0:1] {1};
					\addlegendentry{$f$}
					
					\addplot[red,thick,domain=\RSparamP:\RSparamB]
					{\RSparamP/x^2};
					\addlegendentry{\textcolor{red}{$g$}}
					
					\addplot[red,thick,domain=0:\RSparamP,forget plot]
					{0};
					\addplot[red,thick,domain=\RSparamB:1,forget plot]
					{0};
					
					\addplot[gray,dotted,forget plot]
					coordinates {
						(\RSparamP,-0.10)
						(\RSparamP,4.910815)
					};
					
					\addplot[gray,dotted,forget plot]
					coordinates {
						(\RSparamTau,-0.10)
						(\RSparamTau,1)
					};
					
					\draw[red,thick,->]
					(axis cs:\RSparamB,0)
					--(axis cs:\RSparamB,3.20);
					\node[red,above,font=\footnotesize]
					at (axis cs:\RSparamB,3.20) {$p/B$};
				\end{axis}
			\end{tikzpicture}
		\end{subfigure}\hfill
		\begin{subfigure}[t]{0.32\textwidth}
			\centering
			\subcaption{Prior and target cdfs.}
			\begin{tikzpicture}
				\begin{axis}[appdistaxis,
					ymin=-0.02,ymax=1.10,
					ytick={0,1},
					extra x ticks={\RSparamP,\RSparamTau,\RSparamL,\RSparamB},
					extra x tick labels={$p$,$\tau$,$L$,$B$},
					]
					\addplot[gray,dotted,forget plot]
					coordinates {
						(\RSparamB,-0.02)
						(\RSparamB,1)
					};
					
					\addplot[gray,dotted,forget plot]
					coordinates {
						(\RSparamTau,-0.10)
						(\RSparamTau,\RSparamTau)
					};
					
					\addplot[gray,dotted,forget plot]
					coordinates {
						(\RSparamL,-0.02)
						(\RSparamL,\RSparamL)
					};
					
					\addplot[black,thick,domain=0:1] {x};
					\addlegendentry{$F$}
					
					\addplot[red,thick,domain=0:\RSparamP] {0};
					\addlegendentry{\textcolor{red}{$G$}}
					
					\addplot[red,thick,
					domain=\RSparamP:\RSparamB,forget plot]
					{1-\RSparamP/x};
					
					\addplot[red,thick,forget plot]
					coordinates {
						(\RSparamB,{1-\RSparamP/\RSparamB})
						(\RSparamB,1)
					};
					
					\addplot[red,thick,domain=\RSparamB:1,forget plot]
					{1};
				\end{axis}
			\end{tikzpicture}
		\end{subfigure}\hfill
		\begin{subfigure}[t]{0.32\textwidth}
			\centering
			\subcaption{Integrated-cdfs.}
			\begin{tikzpicture}
				\begin{axis}[appdistaxis,
					ymin=-0.01,ymax=0.55,
					ytick={0,0.2,0.4},
					extra x ticks={\RSparamP,\RSparamTau,\RSparamL,\RSparamB},
					extra x tick labels={$p$,$\tau$,$L$,$B$}
					]
					\addplot[gray,dotted,forget plot]
					coordinates {
						(\RSparamL,-0.01)
						(\RSparamL,{\RSparamL^2/2})
					};
					
					\addplot[gray,dotted,forget plot]
					coordinates {
						(1,-0.01)
						(1,0.5)
					};
					
					\addplot[black,thick,domain=0:1] {x^2/2};
					\addlegendentry{$I_F(\cdot)$}
					
					\addplot[red,thick,domain=0:\RSparamP] {0};
					\addlegendentry{\textcolor{red}{$I_G(\cdot)$}}
					
					\addplot[red,thick,
					domain=\RSparamP:\RSparamB,forget plot]
					{x-\RSparamP-\RSparamP*ln(x/\RSparamP)};
					
					\addplot[red,thick,domain=\RSparamB:1,forget plot]
					{x-0.5};
				\end{axis}
			\end{tikzpicture}
		\end{subfigure}
		\label{Fig: RS_distributions}
		\begin{minipage}{\textwidth}
			\footnotesize
			Note: In panel~(a), $g(s)=p/s^2$ on $[p,B)$ and is zero elsewhere. The red arrow at $B$ denotes the atom of mass $p/B$; its height is schematic and is not a density value. 
			Panel~(b) shows the corresponding cdf jump. 
			In panel~(c), the integrated-cdfs meet only at $0$, $L$, and $1$. 
		\end{minipage}
	\end{figure}
	
	Figure~\ref{Fig: RS_distributions} compares the cdfs and integrated-cdfs of the uniform prior $F$ and buyer-optimal target $G$. 
	The majorization inequality $I_G(\cdot)\leq I_F(\cdot)$ binds at $0$, $L$, and $1$ and is strict between them. 
	Hence, by Proposition~\ref{Prop: Decomposition}, any deterministic signal $T$ that implements $G$ from $F$ must precisely reveal whether the buyer's value $X$ is above or below $L$, and we can separately solve for the implementation schemes on the two non-separating intervals $(0,L)$ and $(L,1)$. 
	
	Consider the first non-separating interval $(0,L)$.
	By \eqref{Eq: RS_cdf} and Figure~\ref{Fig: RS_distributions}, we have
	\begin{equation*}
		\inf\supp\left(G|_{(0,L)}\right) = p > 0
		\qquad \text{and} \qquad
		\sup\supp\left(G|_{(0,L)}\right) = L\ . 
	\end{equation*}
	By part 2 of Proposition~\ref{Prop: Contour_support_gap}, an upper contour nested-interval implementation on $(0,L)$ is infeasible. 
	So let us consider lower contour nested-interval implementation. 
	For $\tau=\sqrt{p}\approx0.45$, we have $g(s)=p/s^2>1=f(s)$ for $p<s<\tau$ and $F(s)\leq G(s^-)=G(s)$ for $\tau<s<L$; see panels (a) and (b) of Figure~\ref{Fig: RS_distributions}. 
	Hence, by part 1 of Proposition~\ref{Prop: Contour_density_cdf}, a lower contour nested-interval implementation is indeed feasible. 
	Moreover, since $F$ is uniform, for each $s\in(p,L)$ the lower contour interval $[a_1(s),b_1(s)]=\{x\in(0,L):\ T(x)\leq s\}$ can be uniquely pinned down by the mass and first moment constraints
	\begin{align*}
		F\left(b_1(s)\right) - F(a_1(s)) = b_1(s)-a_1(s) &= G(s) - G(0) = 1-\frac{p}{s},\\
		\int^{b_1(s)}_{a_1(s)}x\,dF(x) = \frac{b_1(s)^2-a_1(s)^2}{2} &= \int^{s}_{p}m\,dG(m) = p\log\frac{s}{p}.
	\end{align*}
	Solving them yields
	\begin{equation*}
		a_1(s) = \frac{sp}{s-p}\log\frac{s}{p}-\frac{s-p}{2s} 
		\qquad \text{and} \qquad 
		b_1(s) = \frac{sp}{s-p}\log\frac{s}{p}+\frac{s-p}{2s} 
	\end{equation*}
	for all $s\in(p,L]$. Set $a_1(p)=b_1(p)=p$ to make both functions continuous at $p$. 
	It can be verified that $0\leq a_1(s)\leq p\leq b_1(s)\leq L$ for all $s\in[p,L]$, $a_1(s)$ is strictly decreasing, and $b_1(s)$ is strictly increasing. 
	Hence, both functions are invertible and, by \eqref{Eq: Lower_Contour_nested_interval}, any lower contour nested-interval implementation $T$ on $(0,L)$ must a.e. coincide with
	\begin{equation}\label{Eq: RS_implementation_1}
		T_1(x) = \inf\{s\in[p,L]:\ x\in[a_1(s),b_1(s)]\} = 
		\begin{cases}
			a_1^{-1}(x), & \text{if } x\in(0,p) \\
			b_1^{-1}(x), & \text{if } x\in[p,L)\ . 
		\end{cases}
	\end{equation} 
	As $s$ rises from $p$ to $L$, the lower contour interval expands from $\{p\}$ to $[0,L]$. 
	
	Now consider the second non-separating interval $(L,1)$. 
	By \eqref{Eq: RS_cdf} and Figure~\ref{Fig: RS_distributions}, we have
	\begin{equation*}
		\inf\supp\left(G|_{(L,1)}\right) = L 
		\qquad \text{and} \qquad
		\sup\supp\left(G|_{(L,1)}\right) = B<1\ . 
	\end{equation*}
	It follows from part 1 of Proposition~\ref{Prop: Contour_support_gap} that a lower contour nested-interval implementation on $(L,1)$ is infeasible. 
	On the other hand, since $F(s)\geq G(s)$ for $L<s<B$, it follows from part 2 of Proposition~\ref{Prop: Contour_density_cdf} that an upper contour nested-interval implementation is feasible instead. 
	Moreover, for each $s\in[L,B]$ the upper contour interval $[a_2(s),b_2(s)]=\{x\in(L,1):\ T(x)\geq s\}$ can be uniquely pinned down by the mass and first moment constraints
	\begin{align*}
		F\left(b_2(s)\right) - F(a_2(s)) = b_2(s)-a_2(s) &= G(1) - G(s^{-}) = \frac{p}{s}, \\
		\int^{b_2(s)}_{a_2(s)}x\,dF(x) = \frac{b_2(s)^2-a_2(s)^2}{2} &= \int_{[s,B]}{m\,dG(m)} = \frac{1}{2}-p\log\frac{s}{p}.
	\end{align*}
	Solving them yields
	\begin{equation*}
		a_2(s)=\frac12\left(\frac{s}{p}-\frac{p}{s}\right)-s\log\frac{s}{p}
		\qquad \text{and} \qquad 
		b_2(s)=\frac12\left(\frac{s}{p}+\frac{p}{s}\right)-s\log\frac{s}{p}
	\end{equation*}
	for all $s\in[L,B]$. 
	These intervals shrink from $[L,1]$ to $\left[\bar{a},\bar{b}\right]$, where 
	\begin{equation*}
		\bar{a} := a_2(B)=B-\frac{p}{2B}
		\qquad \text{and} \qquad 
		\bar{b} := b_2(B)=B+\frac{p}{2B}\ .
	\end{equation*}
	The terminal interval has prior mass $p/B$ and conditional mean $B$, and hence generates the target atom. 
	Hence, by \eqref{Eq: Upper_Contour_nested_interval}, any upper contour nested-interval implementation $T$ on $(L,B)$ must a.e. coincide with
	\begin{equation}\label{Eq: RS_implementation_2}
		T_2(x) = \sup\{s\in[L,B]:\ x\in[a_2(s),b_2(s)]\} = 
		\begin{cases}
			a_2^{-1}(x), & \text{if } x\in[L,\bar{a}) \\
			B, & \text{if } x\in\left[\bar{a},\bar{b}\right] \\
			b_2^{-1}(x), & \text{if } x\in(\bar{b},1]\ . 
		\end{cases}
	\end{equation} 
	Finally, by combining $T_1$ and $T_2$ from \eqref{Eq: RS_implementation_1} and \eqref{Eq: RS_implementation_2}, we obtain a contour nested-interval implementation $T$ that implements $G$ from $F$; it coincides with $T_1$ on $(0,L)$ and coincides with $T_2$ on $(L,1)$. 
	Figure~\ref{Fig: RS_maps} illustrates this implementation.
	
	\begin{figure}[htbp]
		\centering
		\caption{A contour nested-interval implementation of the Roesler--Szentes optimal distribution}
		\begin{tikzpicture}
			\begin{axis}[
				width=0.6\textwidth,
				height=0.25\textheight,
				axis lines=left,
				axis line style={line width=0.8pt},
				xmin=0,xmax=1.025,
				ymin=0,ymax=1.05,
				xtick={0,0.203632,0.715332,1},
				xticklabels={$0$,$p$,$L$,$1$},
				ytick={0,0.203632,0.715332,0.872814,1},
				yticklabels={$0$,$p$,$L$,$B$,$1$},
				xlabel={$x$},
				xlabel style={at={(axis description cs:1.03,-0.05)},anchor=south,rotate=0},
				ylabel={$T(x)$},
				ylabel style={at={(axis description cs:0,1.0)},anchor=south,rotate=-90},
				label style={font=\small},
				tick label style={font=\small},
				clip=false,
				samples=160
				]
				\addplot[black,very thick,domain=0.2037:0.715332,variable=\s]
				({\s*ln(\s/0.203632)/(\s/0.203632-1)-(1-0.203632/\s)/2},{\s});
				\addplot[black,very thick,domain=0.2037:0.715332,variable=\s]
				({\s*ln(\s/0.203632)/(\s/0.203632-1)+(1-0.203632/\s)/2},{\s});
				\addplot[black,very thick,domain=0.715332:0.872814,variable=\s]
				({0.5*(\s/0.203632-0.203632/\s)-\s*ln(\s/0.203632)},{\s});
				\addplot[black,very thick,domain=0.715332:0.872814,variable=\s]
				({0.5*(\s/0.203632+0.203632/\s)-\s*ln(\s/0.203632)},{\s});
				\addplot[black,very thick,domain=0.756162:0.989467]
				{0.872814};
				\addplot[gray,dashed] coordinates {(0,0.715332) (1,0.715332)};
				\addplot[gray,dashed] coordinates {(0,0.872814) (1,0.872814)};
				\node[above right,yshift=2pt,font=\small] at (axis cs:0.756162,0.872814) {$\bar a$};
				\node[above left,yshift=2pt,font=\small] at (axis cs:0.989467,0.872814) {$\bar b$};
			\end{axis}
		\end{tikzpicture}
		\label{Fig: RS_maps}
	\end{figure}
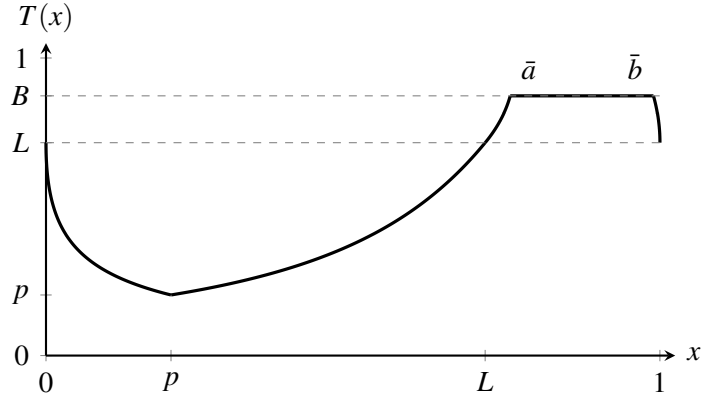

	\subsection{Robust persuasion with ambiguous receiver preferences}\label{Subsec: SapiroGheiler2024}
	In \citet{SapiroGheiler2024}, a sender commits to an experiment about a state $X\in[0,1]$ with prior cdf $F$.
	After observing the signal, a receiver chooses a binary action $a\in\{0,1\}$.
	The receiver has a private type $R\in[0,1]$, which is independent of the state. Given $X$ and $R$, the receiver's payoff is $a(X-R)$ and the sender's payoff is $a$.
	Ties are broken against the sender, so the receiver chooses $a=1$ if and only if the posterior mean $M$ strictly exceeds $R$.
	The sender knows only that the distribution $\rho$ of $R$ is supported on $[0,1]$ and has mean $r^\ast$.
	The sender's max-min persuasion problem is thus
	\begin{equation}\label{Eq: Problem_SapiroGheiler}
		\max_{H\in\MPC(F)}
		\min_{\substack{\rho\in\Delta([0,1])\\
				\int_0^1 r\,d\rho(r)=r^\ast}}
		\int_{[0,1]}\int_{[0,1]}
		\1_{\{m>r\}}\,dH(m)\,d\rho(r).
	\end{equation}
	That is, the sender chooses a feasible posterior-mean distribution given prior $F$, while an adversarial nature chooses the cutoff distribution $\rho$.
	
	Suppose that $F$ is atomless, non-uniform, and concave, with a continuous and strictly positive density function $f$ on $(0,1)$. 
	Let $\pi:=\E_F[X]$, so $0<\pi<1/2$, and suppose $0<r^\ast\leq\pi$.
	\citet[Section~5.1]{SapiroGheiler2024} shows that the sender's unique robustly optimal solution $G$ to \eqref{Eq: Problem_SapiroGheiler} is a uniform distribution on $[0,2\pi]$, under which the sender's worst-case success probability is $1-r^\ast/{2\pi}$. 
	Using results in Sections~\ref{Sec: Construction} and \ref{Sec: Contour_nested_interval}, we show that this robustly optimal solution admits an upper contour nested-interval implementation. 
	
	\begin{proposition}
		\label{Prop: Contour_uniform_target}
		Let $F$ be an atomless, non-uniform, and concave cdf on $[0,1]$, with a continuous density $f>0$ on $(0,1)$.
		Write $\pi:=\E_F[X]$ and let $G$ be a uniform distribution on $[0,2\pi]$.
		Then $G$ admits a unique upper contour nested-interval implementation, up to $F$-null sets.
		Its upper contour intervals \([a(s),b(s)]\) are uniquely determined, for $0<s<2\pi$, by 
		\begin{align}
			F\left(b(s)\right)-F\left(a(s)\right) &= G(2\pi)-G(s)=1-\frac{s}{2\pi}, \label{Eq: SG_contour_mass}\\
			\int_{a(s)}^{b(s)}x\,dF(x) &= \int_s^{2\pi}m\,dG(m)=\pi-\frac{s^2}{4\pi}. \label{Eq: SG_contour_first_moment}
		\end{align}
		As $s$ rises from $0$ to $2\pi$, $a(s)$ increases continuously and strictly from $0$ to $2\pi$, while $b(s)$ decreases continuously and strictly from $1$ to $2\pi$. 
		The implementation is given by 
		\begin{equation}\label{Eq: SG_implementation}
			T(x)=\sup\{s\in[0,2\pi]:\ x\in[a(s),b(s)]\}
			=\begin{cases}
				a^{-1}(x),&\text{if }x\in[0,2\pi],\\
				b^{-1}(x),&\text{if }x\in(2\pi,1].
			\end{cases}
		\end{equation}
	\end{proposition}
	\begin{proof}[Proof of Proposition~\ref{Prop: Contour_uniform_target}]
		See Appendix~\ref{App. Proof_Contour_uniform_target}. 
	\end{proof}
	
	Figure~\ref{Fig: SG_maps} illustrates the geometry of $T(\cdot)$ given by \eqref{Eq: SG_implementation}: 
	the upper contour interval $[a(s),b(s)]$ shrinks from $[0,1]$ to $\{2\pi\}$ as $s$ rises from $0$ to $2\pi$. 
	
	\begin{figure}[H]
		\centering
		\caption{An upper contour nested-interval implementation under a concave prior $F$}
		\begin{tikzpicture}
			\pgfmathsetmacro{\SGillustrationT}{1/3}
			\pgfmathsetmacro{\SGillustrationD}{sqrt(2*\SGillustrationT^2+4*\SGillustrationT+3)}
			\pgfmathsetmacro{\SGillustrationA}{1+\SGillustrationT+\SGillustrationT^2+\SGillustrationT*\SGillustrationD}
			\pgfmathsetmacro{\SGillustrationLower}{1-(1+\SGillustrationT)/\SGillustrationA}
			\pgfmathsetmacro{\SGillustrationUpper}{1-\SGillustrationT*(1+\SGillustrationT)/\SGillustrationA}
			\pgfmathsetmacro{\SGillustrationMean}{4*\SGillustrationT*\SGillustrationD/(3*\SGillustrationA)}
			\begin{axis}[
				width=0.6\textwidth,
				height=0.25\textheight,
				axis lines=left,
				axis line style={line width=0.8pt},
				xmin=0,xmax=1.025,
				ymin=0,ymax=0.71,
				xtick={0,\SGillustrationLower,2/3,\SGillustrationUpper,1},
				xticklabels={$0$,$a(s)$,$2\pi$,$b(s)$,$1$},
				ytick={0,\SGillustrationMean,2/3},
				yticklabels={$0$,$s$,$2\pi$},
				xlabel={$x$},
				xlabel style={at={(axis description cs:1.03,-0.05)},anchor=south,rotate=0},
				ylabel={$T(x)$},
				ylabel style={at={(axis description cs:0,1.0)},anchor=south,rotate=-90},
				label style={font=\small},
				tick label style={font=\small},
				clip=false,
				samples=220
				]
				\addplot[black,very thick,domain=0:1,variable=\t]
				({
					1-(1+\t)/
					(1+\t+\t^2+\t*sqrt(2*\t^2+4*\t+3))
				},{
					4*\t*sqrt(2*\t^2+4*\t+3)/
					(3*(1+\t+\t^2+\t*sqrt(2*\t^2+4*\t+3)))
				});
				\addplot[black,very thick,domain=0:1,variable=\t]
				({
					1-\t*(1+\t)/
					(1+\t+\t^2+\t*sqrt(2*\t^2+4*\t+3))
				},{
					4*\t*sqrt(2*\t^2+4*\t+3)/
					(3*(1+\t+\t^2+\t*sqrt(2*\t^2+4*\t+3)))
				});
				\draw[gray,dashed]
				(axis cs:2/3,0) -- (axis cs:2/3,2/3);
				\draw[gray,dashed]
				(axis cs:0,\SGillustrationMean) -- (axis cs:1,\SGillustrationMean);
				\draw[gray,dotted]
				(axis cs:\SGillustrationLower,0) -- (axis cs:\SGillustrationLower,\SGillustrationMean);
				\draw[gray,dotted]
				(axis cs:\SGillustrationUpper,0) -- (axis cs:\SGillustrationUpper,\SGillustrationMean);
				\node[circle,fill=black,inner sep=1.2pt]
				at (axis cs:\SGillustrationLower,\SGillustrationMean) {};
				\node[circle,fill=black,inner sep=1.2pt]
				at (axis cs:\SGillustrationUpper,\SGillustrationMean) {};
			\end{axis}
		\end{tikzpicture}
		\label{Fig: SG_maps}
	\end{figure}
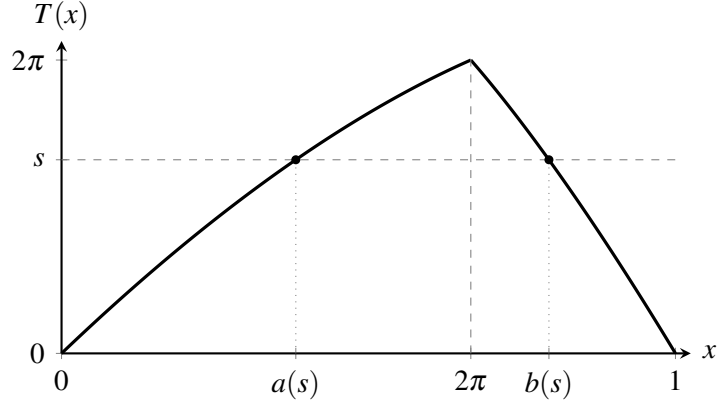
	
	To understand why Proposition~\ref{Prop: Contour_uniform_target} holds, consider the triangular prior $F(x)=2x-x^2$ on $[0,1]$, which is concave and admits a strictly decreasing density $f$. 
	In this case, $\E_F[X]=\pi=1/3$ so the robustly optimal $G$ is a uniform distribution on $[0,2\pi]=[0,2/3]$. 
	Figure~\ref{Fig: SG_distributions} depicts the densities, cdfs, and integrated-cdfs of this $F$ and $G$. 
	
	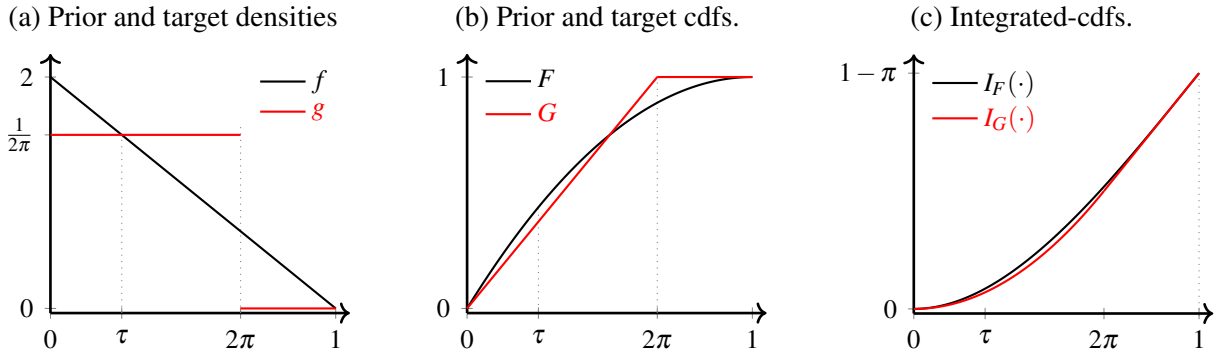
\begin{figure}[htbp]
		\centering
		\caption{The triangular prior $F$ and the robustly optimal $G$}
		
		\pgfmathsetmacro{\SGcutoff}{2/3}
		
		\pgfplotsset{appdistaxis/.style={
				width=0.75\linewidth,
				height=0.15\textheight,
				scale only axis,
				axis lines=left,
				axis line style={line width=1pt,->},
				xmin=0,xmax=1.05,
				xtick={0,0.25,\SGcutoff,1},
				xticklabels={$0$,$\tau$,$2\pi$,$1$},
				tick label style={font=\footnotesize},
				axis background/.style={fill=none},
				legend style={
					at={(0.03,0.98)},anchor=north west,
					font=\footnotesize,
					legend cell align=left,
					fill=none,draw=none
				},
				clip=false,
				samples=160
		}}
		
		\begin{subfigure}[t]{0.32\textwidth}
			\centering
			\subcaption{Prior and target densities}
			\begin{tikzpicture}
				\begin{axis}[appdistaxis,
					ymin=-0.04,ymax=2.20,
					ytick={0,1.5,2},
					yticklabels={$0$,$\frac{1}{2\pi}$,$2$},
					legend style={
						at={(0.97,0.98)},anchor=north east
					}
					]
					\addplot[gray,dotted,forget plot]
					coordinates {
						(\SGcutoff,-0.04)
						(\SGcutoff,1.60)
					};
					
					\addplot[gray,dotted,forget plot]
					coordinates {
						(0.25,-0.04)
						(0.25,1.5)
					};
					
					\addplot[black,thick,domain=0:1] {2*(1-x)};
					\addlegendentry{$f$}
					
					\addplot[red,thick,domain=0:\SGcutoff] {1.5};
					\addlegendentry{\textcolor{red}{$g$}}
					
					\addplot[red,thick,domain=\SGcutoff:1,forget plot]
					{0};
				\end{axis}
			\end{tikzpicture}
		\end{subfigure}\hfill
		\begin{subfigure}[t]{0.32\textwidth}
			\centering
			\subcaption{Prior and target cdfs.}
			\begin{tikzpicture}
				\begin{axis}[appdistaxis,
					ymin=-0.02,ymax=1.10,
					ytick={0,1}
					]
					\addplot[gray,dotted,forget plot]
					coordinates {
						(\SGcutoff,-0.02)
						(\SGcutoff,1.02)
					};
					
					\addplot[gray,dotted,forget plot]
					coordinates {
						(0.25,-0.04)
						(0.25,0.4375)
					};
					
					\addplot[black,thick,domain=0:1] {2*x-x^2};
					\addlegendentry{$F$}
					
					\addplot[red,thick,domain=0:\SGcutoff] {1.5*x};
					\addlegendentry{\textcolor{red}{$G$}}
					
					\addplot[red,thick,domain=\SGcutoff:1,forget plot]
					{1};
				\end{axis}
			\end{tikzpicture}
		\end{subfigure}\hfill
		\begin{subfigure}[t]{0.32\textwidth}
			\centering
			\subcaption{Integrated-cdfs.}
			\begin{tikzpicture}
				\begin{axis}[appdistaxis,
					ymin=-0.013,ymax=0.72,
					ytick={0,\SGcutoff},
					yticklabels={$0$,$1-\pi$}
					]
					\addplot[gray,dotted,forget plot]
					coordinates {
						(1,-0.02)
						(1,0.667)
					};
					
					\addplot[black,thick,domain=0:1] {x^2-x^3/3};
					\addlegendentry{$I_F(\cdot)$}
					
					\addplot[red,thick,domain=0:\SGcutoff]
					{0.75*x^2};
					\addlegendentry{\textcolor{red}{$I_G(\cdot)$}}
					
					\addplot[red,thick,domain=\SGcutoff:1,forget plot]
					{x-1/3};
				\end{axis}
			\end{tikzpicture}
		\end{subfigure}
		\label{Fig: SG_distributions}
		\begin{minipage}{\textwidth}
			\footnotesize
			Note: The prior is triangular with cdf $F(x)=2x-x^2$. The robustly optimal solution $G$ is uniform on $[0,2\pi]=[0,2/3]$. 
			In panel (c), the integrated-cdf inequality is strict throughout $(0,1)$, so this is the unique non-separating interval. 
		\end{minipage}
	\end{figure}
	
	For this pair $(F,G)$, the majorization inequality $I_G(\cdot)\leq I_F(\cdot)$ is strict on $(0,1)$, and hence this is only one non-separating interval on which the target support satisfies
	\begin{equation*}
		\inf\supp(G)=0
		\qquad \text{and} \qquad
		\sup\supp(G)=2\pi<1.
	\end{equation*}
	It follows from part 1 of Proposition~\ref{Prop: Contour_support_gap} that a lower contour nested-interval implementation is infeasible.
	Meanwhile, for $\tau=1/4\in(0,2\pi)$, it holds that $F(x)\geq G(x)$ for all $x\in[0,\tau]$ and $f(x)\leq g(x)$ for all $x\in[\tau,2\pi]$; see panels (a) and (b) of Figure~\ref{Fig: SG_distributions}. 
	Hence, by part 2 of Proposition~\ref{Prop: Contour_density_cdf}, an upper contour nested-interval implementation is feasible, and it is uniquely given by \eqref{Eq: SG_implementation} up to $F$-null sets. 
	These features generalize to every prior $F$ that satisfies the assumptions of Proposition~\ref{Prop: Contour_uniform_target}. 
	
	\section{Conclusion}\label{Sec}
	In this paper we study whether randomization is needed to implement a prescribed distribution of posterior means and how a deterministic implementation can be constructed. For a one-dimensional, atomless prior with a finite first moment, the backward-Monge theorem of \cite{Nutz2024martingale} implies that every mean-preserving contraction admits a deterministic implementation. We describe a construction that decomposes the problem into independent local problems on non-separating intervals and uses barcode transport within each interval. We further characterize when the implementing signal can have nested lower or upper contour intervals and provide conditions on the prior and target that help determine which direction to use.
	
	These results imply that restricting information design to deterministic signals leaves the designer's value unchanged whenever the objective and constraints depend only on the posterior-mean distribution. This result applies beyond linear persuasion, including settings with endogenous pricing or worst-case evaluations, as our applications illustrate. A designer can therefore first choose a desirable posterior-mean distribution and then implement it through a partition of the state space. The resulting partition may nevertheless pool disconnected states and need not be monotone, so additional restrictions on the structure of disclosure can remain costly. Our nested-interval characterization identifies circumstances in which the desired outcome can be achieved through a simpler deterministic rule.

	\begin{appendices}
		\section{All Omitted Proofs}\label{App: Proofs}
		
		\subsection{Proof of Proposition~\ref{Prop: Decomposition}}\label{App. Prop_Decomposition}
		We first introduce a useful auxiliary result. 
		\begin{claim}\label{Claim: Integrated_cdf_identities}
			Let $F$ be atomless and $G\in\MPC(F)$. 
			Suppose that random vector $(X,M)$ satisfies $X\sim F$, $M\sim G$, and $\E[X|M]=M$. 
			Then, for every $z\in\mathbb R$, we have
			\begin{equation*}
				I_F(z) = \E\left[(z-X)\1_{\{X\leq z\}}\right]
				\qquad \text{and} \qquad 
				I_G(z) = \E\left[(z-X)\1_{\{M\leq z\}}\right].
			\end{equation*}
			Consequently,
			\begin{align}\label{Eq: Difference_integrated_cdfs}
				\begin{split}
					I_F(z)-I_G(z) &= \E\left[(z-X)\left(\1_{\{X\le z\}}-\1_{\{M\le z\}}\right)\right]\\
					&= \E\left[(z-X)\1_{\{X\leq z\}}\1_{\{M>z\}}+(X-z)\1_{\{X>z\}}\1_{\{M\le z\}}\right] \geq 0.
				\end{split} 
			\end{align} 
		\end{claim}
		\begin{proof}[Proof of Claim~\ref{Claim: Integrated_cdf_identities}]
			For any integrable random variable $Y$ with cdf $H$ and a cutoff $z\in\R$, integration by parts gives
			\begin{equation}\label{Eq: I_z_identity}
				\E[(z-Y)\1_{\{Y\leq z\}}] = \int_{-\infty}^{z}{(z-y)dH(y)} = \int_{-\infty}^{z}{H(y)dy} = I_H(z)\ .
			\end{equation} 
			Applying \eqref{Eq: I_z_identity} to $X$ gives the expression for $I_F(z)$. 
			Applying it to $M$ and using $\E[X|M]=M$ gives
			\begin{equation*}
				I_G(z) = \E\left[(z-M)\1_{\{M\leq z\}}\right]
				=\E\left[\E[z-X|M]\1_{\{M\leq z\}}\right]
				=\E\left[(z-X)\1_{\{M\leq z\}}\right].
			\end{equation*}
			The first line of \eqref{Eq: Difference_integrated_cdfs} then directly follows from the expressions of $I_F(z)$ and $I_G(z)$. 
			The second line follows from the fact that
			\begin{align*}
				(z-X)\left(\1_{\{X\leq z\}}-\1_{\{M\leq z\}}\right) 
				&= (z-X)\1_{\{X\leq z<M\}} + (X-z)\1_{\{M\leq z<X\}} \\
				&= (z-X)\1_{\{X\leq z\}}\1_{\{M>z\}} + (X-z)\1_{\{X>z\}}\1_{\{M\le z\}}.
			\end{align*}
			Since both terms are non-negative for all $z$, we have $I_F(z)-I_G(z)\geq0$. 
		\end{proof}
		
		\noindent Let $X\sim F$ and $M=T(X)$. Since $T$ implements $G$ from $F$, we have $M\sim G$ and $\mathbb{E}[X|M]=M$. 
		
		To show part 1 of Proposition~\ref{Prop: Decomposition}, recall that $I_F(z)-I_G(z)=0$ for all $z\in\mathbb{S}_G$ by definition. 
		Then, by \eqref{Eq: Difference_integrated_cdfs} in Claim~\ref{Claim: Integrated_cdf_identities}, we have 
		\begin{equation}\label{Eq: Difference_integrated_cdfs2}
			I_F(z)-I_G(z) = \E\left[(z-X)\1_{\{X\leq z\}}\1_{\{M>z\}} + (X-z)\1_{\{X>z\}}\1_{\{M\le z\}}\right] = 0 \quad \text{for all } z\in\mathbb{S}_G\ . 
		\end{equation}
		Since $F$ is atomless, \eqref{Eq: Difference_integrated_cdfs2} implies $\1_{\{X\leq z\}}=\1_{\{M\leq z\}}$ and $\mathbb{P}(X=z)=\mathbb{P}(M=z)=0$ for all $z\in\mathbb{S}_G$. 
		Using $M=T(X)$, we then have $\1_{\{X\leq z\}}=\1_{\{T(X)\leq z\}}$ for a.e. $X$ and $F(z)=G(z)=G(z^{-})$ for every $z\in\mathbb{S}_G$. 
		Therefore, for any measurable subset $D$ of $\mathbb{S}_G$, we have $\1_{\{X\in D\}}=\1_{\{T(X)\in D\}}$ for a.e. $X\in D$. Consequently, $T(X)=X$ a.e. on $\mathbb{S}_G$. 
		This proves part 1 of Proposition~\ref{Prop: Decomposition}.  
		
		To show part 2, consider any non-separating interval $(c_j,d_j)$. 
		Since $c_j,d_j\in\mathbb{S}_G$, the argument above implies $F(c_j)=G(c_j)$ and $F(d_j)=G(d_j)$. 
		Hence, $F(d_j)-F(c_j)=G(d_j)-G(c_j)$ and we denote this value by $p_j$. 
		Meanwhile, recall that 
		\begin{equation*}
			F|_{(c_j,d_j)}(x) = 
			\begin{cases}
				0, &\text{if } x<c_j \\
				\frac{F(x)-F(c_j)}{F(d_j)-F(c_j)}, &\text{if } x\in(c_j,d_j] \\
				1, &\text{if } x\geq d_j
			\end{cases}
			\quad \text{and} \quad
			G|_{(c_j,d_j)}(x) = 
			\begin{cases}
				0, &\text{if } x<c_j \\
				\frac{G(x)-G(c_j)}{G(d_j)-G(c_j)}, &\text{if } x\in(c_j,d_j] \\
				1, &\text{if } x\geq d_j
			\end{cases}\ . 
		\end{equation*}
		These together imply 
		\begin{equation*}
			I_{F|_{(c_j,d_j)}}(z) - I_{G|_{(c_j,d_j)}}(z) = \frac{I_F(z) - I_G(z)}{p_j}\cdot\1_{z\in(c_j,d_j)} \geq 0 
		\end{equation*}
		for all $z\in\R$, and strict inequality holds if and only if $z\in(c_j,d_j)$. 
		Therefore, by \cite{Strassen1965existence}, $G|_{(c_j,d_j)}\preceq_{MPS}F|_{(c_j,d_j)}$. 
		This proves part 2 of Proposition~\ref{Prop: Decomposition}.

		\subsection{Proof of Lemma~\ref{Lemma: Moment_matched_interval}}\label{App. Proof_Lemma_Interval}
		Let $\mu$ and $\nu$ be the local target and prior measures from Section~\ref{Sec: Contour_nested_interval}. By Proposition~\ref{Prop: Decomposition}, $\mu\leq_{cx}\nu$, and $\nu$ is atomless. For $0\leq\alpha\leq\mu$, write
		\begin{equation*}
			h:=\alpha(\R)\qquad\text{and}\qquad u:=\int_{\R}m\,d\alpha(m).
		\end{equation*}
		If $h=0$, take $I=\varnothing$. If $h=1$, then $\alpha=\mu$, and the entire component has the required mass and first moment. In either case, every submeasure of $\nu$ with the required mass is the same measure, so uniqueness and minimality are immediate. Assume henceforth that $0<h<1$.
		
		For $y\in(0,1)$, let
		\begin{equation*}
			Q(y):=\inf\{x\in\R:\nu((-\infty,x])\geq y\}
		\end{equation*}
		be the quantile function of $\nu$, and define $Q(0)$ and $Q(1)$ by their one-sided limits. Either endpoint may be infinite. Since $\nu$ is atomless, its restriction to $[Q(z),Q(z+h)]$ has mass $h$ and first moment
		\begin{equation*}
			J(z):=\int_z^{z+h}Q(y)\,dy,\qquad 0\leq z\leq1-h.
		\end{equation*}
		The finite first moment makes $Q$ integrable and $J$ continuous. Atomlessness also makes $Q$ strictly increasing on $(0,1)$, so $J$ is strictly increasing.
		
		We next show that $J(0)\leq u\leq J(1-h)$. For any $q\in\R$, the inequalities $\alpha\leq\mu$ and $\mu\leq_{cx}\nu$ give
		\begin{align*}
			hq-u&\leq\int_{\R}(q-m)_+\,d\alpha(m)
			\leq\int_{\R}(q-m)_+\,d\mu(m)
			\leq\int_{\R}(q-x)_+\,d\nu(x),\\
			u-hq&\leq\int_{\R}(m-q)_+\,d\alpha(m)
			\leq\int_{\R}(m-q)_+\,d\mu(m)
			\leq\int_{\R}(x-q)_+\,d\nu(x),
		\end{align*}
		where $v_+:=\max\{v,0\}$. Taking $q=Q(h)$ in the first line and $q=Q(1-h)$ in the second yields
		\begin{equation*}
			u\geq\int_0^hQ(y)\,dy=J(0),
			\qquad
			u\leq\int_{1-h}^1Q(y)\,dy=J(1-h).
		\end{equation*}
		Consequently, there is a unique $z^*\in[0,1-h]$ such that $J(z^*)=u$. Set
		\begin{equation*}
			a:=Q(z^*),\qquad b:=Q(z^*+h),\qquad
			I:=[a,b]\cap(c_j,d_j).
		\end{equation*}
		Then $\nu|_I$ has mass $h$ and first moment $u$. Every interval restriction of mass $h$ corresponds to a quantile interval of length $h$, so uniqueness of $z^*$ gives uniqueness of the restriction up to $\nu$-null sets.
		
		It remains to prove minimality in convex order. Write $\beta:=\nu|_I$, and let $\eta\leq\nu$ have the same mass and first moment as $\beta$. Suppose first that $a$ and $b$ are finite. For any convex function $\varphi$, let $L$ be the affine function joining $(a,\varphi(a))$ and $(b,\varphi(b))$. Convexity gives $\varphi-L\leq0$ on $[a,b]$ and $\varphi-L\geq0$ outside it. Since $\eta\leq\nu$,
		\begin{equation*}
			\int_{\R}(\varphi-L)\,d\eta
			\geq\int_{[a,b]}(\varphi-L)\,d\nu
			=\int_{\R}(\varphi-L)\,d\beta.
		\end{equation*}
		The affine terms have equal integrals because mass and first moment agree. Hence $\int\varphi\,d\beta\leq\int\varphi\,d\eta$, proving $\beta\leq_{cx}\eta$.
		
		If $a=-\infty$ and $b<\infty$, equal mass and first moment imply
		\begin{equation*}
			\int_{\R}(x-b)\,d(\eta-\beta)(x)=0.
		\end{equation*}
		Below $b$, both $x-b$ and $\eta-\beta=\eta-\nu$ are non-positive; above $b$, both are non-negative. The integral can therefore vanish only if $\eta=\beta$, since $\nu$ has no atom at $b$. The case $b=+\infty$ is symmetric. This proves minimality for unbounded intervals as well.

		\subsection{Proof of Proposition~\ref{Prop: Nested_Contours}}\label{App. Proof_Prop_Nested_Contours}
		We prove (a)$\Leftrightarrow$(b), (c)$\Leftrightarrow$(d), and (b)$\Leftrightarrow$(c) in sequence. Throughout, $\mu$ and $\nu$ denote the probability measures of the local target and prior, respectively, and $\alpha_s=\mu|_{(c_j,s]}$, $\beta_s=\nu|_{I_j(s)}$.
		
		\paragraph{(a)$\Leftrightarrow$(b):}
		Suppose (a) holds. Let $X\sim\nu$, $M=T(X)$, and $A_s:=\{x\in(c_j,d_j):T(x)\leq s\}$. Since $T$ implements $\mu$ from $\nu$, $M\sim\mu$ and $\E[X|M]=M$. Consequently,
		\begin{align*}
			\nu(A_s)&=\mu((c_j,s]),\\
			\int_{A_s}x\,d\nu(x)&=\E[M\1_{\{M\leq s\}}]
			=\int_{(c_j,s]}m\,d\mu(m).
		\end{align*}
		By (a), $A_s$ is an interval up to $\nu$-null sets. The uniqueness in Lemma~\ref{Lemma: Moment_matched_interval} therefore implies $\nu|_{A_s}=\beta_s$. Since $A_s\subseteq A_t$ whenever $s<t$, (b) follows.
		
		Conversely, suppose (b) holds. First replace each $I_j(s)$ by $\bigcap_{q\in\mathbb Q\cap(s,d_j)}I_j(q)$. By (b) and $\mu((c_j,q])\to\mu((c_j,s])$ as $q\downarrow s$, this leaves its $\nu$-restriction unchanged. Continue to denote these intervals by $I_j(s)$. They are nested and satisfy $I_j(s)=\bigcap_{s<t<d_j}I_j(t)$. Define
		\begin{equation*}
			T_L(x):=\inf\{s\in(c_j,d_j):x\in I_j(s)\}.
		\end{equation*}
		The infimum defines a Borel function, since its strict lower level sets are countable unions of $I_j(q)$ at rational levels $q$. The matching masses tend to zero and one as $s$ approaches $c_j$ and $d_j$, respectively. Thus, this infimum lies in $(c_j,d_j)$ for $\nu$-almost every $x$; assign a fixed value in this interval on the remaining null set. The construction gives $\{x:T_L(x)\leq s\}=I_j(s)$ up to $\nu$-null sets. 
		For $X\sim\nu$ and $M=T_L(X)$, the matching mass and first moment give
		\begin{align*}
			\Prob(M\leq s)&=\nu(I_j(s))=\mu((c_j,s]),\\
			\E[X\1_{\{M\leq s\}}]&=\int_{I_j(s)}x\,d\nu(x)
			=\int_{(c_j,s]}m\,d\mu(m).
		\end{align*}
		The first identity gives $M\sim\mu$. The second then implies $\E[(X-M)\1_{\{M\leq s\}}]=0$ for every $s$. This equality extends from lower intervals to all Borel sets of $M$, so $\E[X|M]=M$. 
		Hence, $T_L$ is a lower contour nested-interval implementation, proving (a). Any other such implementation has the same lower contour restrictions. Comparing these at rational levels shows that it agrees with $T_L$, $\nu$-almost surely.
		
		\paragraph{(c)$\Leftrightarrow$(d):}
		Statement (d) implies (c) by the definition of the shadow. Conversely, suppose (c) holds and fix $s$. Then $\alpha_s\leq_{cx}\beta_s\leq\nu$. Every measure $\eta$ satisfying $\alpha_s\leq_{cx}\eta\leq\nu$ has the same mass and first moment as $\beta_s$. Lemma~\ref{Lemma: Moment_matched_interval} therefore gives $\beta_s\leq_{cx}\eta$. Thus $\beta_s$ is the least element in convex order among the measures defining $S^\nu(\alpha_s)$, proving (d).
		
		\paragraph{(b)$\Leftrightarrow$(c):}
		Suppose (b) holds. By the first equivalence, the constructed map $T_L$ implements $\mu$ from $\nu$. Let $X\sim\nu$ and $M=T_L(X)$. For every convex function $\varphi:\R\to\R$,
		\begin{equation*}
			\int_{\R}\varphi(m)\,d\alpha_s(m) = \E\left[\varphi(M)\1_{\{M\leq s\}}\right] 
			\leq \E\left[\varphi(X)\1_{\{M\leq s\}}\right]
			= \int_{\R}\varphi(x)\,d\beta_s(x). 
		\end{equation*}
		The inequality follows from $\E[X|M]=M$ and the convexity of $\varphi$. Hence, $\alpha_s\leq_{cx}\beta_s$, proving (c).
		
		Conversely, suppose (c) holds. By the preceding equivalence, (d) holds. The associativity property of shadows \citep[][Theorem~4.8]{BeiglbockJuillet2016} gives, for $s<t$,
		\begin{equation*}
			\beta_t=\beta_s+S^{\nu-\beta_s}\!\left(\mu|_{(s,t]}\right)\geq\beta_s.
		\end{equation*}
		Therefore, $\nu(I_j(s)\setminus I_j(t))=0$, proving (b).
		
		Finally, under these equivalent conditions, the joint distribution of $(T_L(X),X)$ matches every lower target contour $\alpha_s$ to its shadow $\beta_s$. It is therefore the left-curtain coupling. Thus the unique lower contour nested-interval implementation is obtained by reading this coupling backward.

		\subsection{Proof of Proposition~\ref{Prop: Contour_support_gap}}\label{App. Proof_Contour_support_gap}
		We only prove part 1; the proof of part 2 is similar and omitted.
		Fix a non-separating interval $(c_j,d_j)$, and retain its local target and prior measures $\mu,\nu$ and common mass $p_j$. 
		By the definition of a non-separating interval and Proposition~\ref{Prop: Decomposition},
		\begin{equation}\label{Eq: Local_strict_gap}
			\int_{\R}(z-x)_+\,d\nu(x)
			-\int_{\R}(z-m)_+\,d\mu(m)
			=\frac{I_F(z)-I_G(z)}{p_j}>0,
			\qquad c_j<z<d_j,
		\end{equation}
		where $v_+:=\max\{v,0\}$.
		
		Let $X\sim\nu$ and $T$ implement $\mu$ from $\nu$, so $T(X)\sim\mu$ and $\E[X|T(X)]=T(X)$.
		Suppose $T$ is a lower contour nested-interval implementation but $\ell_j=c_j$.
		Then $\mu((c_j,s])>0$ for every $s\in(c_j,d_j)$, while this probability tends to zero as $s\downarrow c_j$.
		We can therefore choose $s$ such that the lower contour set
		\begin{equation*}
			E_s:=\{x\in(c_j,d_j):T(x)\leq s\}
		\end{equation*}
		satisfies $0<\nu(E_s)<1$.
		Write $E_s=[a_s,b_s]\cap(c_j,d_j)$, up to $\nu$-null sets, with $c_j\leq a_s<b_s\leq d_j$.
		
		On $\{X\in E_s\}$, we have $X\geq a_s$.
		Since this event is determined by $T(X)$, conditioning on $T(X)$ gives $T(X)\geq a_s$ on the same event whenever
		$a_s$ is finite. 
		If $a_s>c_j$, choose $t\in(c_j,\min\{a_s,s\})$.
		The assumption $\ell_j=c_j$ implies $\Prob(T(X)\leq t)>0$.
		But this event is contained in $\{X\in E_s\}$, where $T(X)\geq a_s>t$, a contradiction.
		Thus $a_s=c_j$, including when $c_j=-\infty$. 
		Consequently, $E_s=(c_j,b_s]$ up to null sets, with a finite $b_s\in(c_j,d_j)$. 
		
		On $\{X\in E_s\}$, $X\leq b_s$; on its complement, $X\geq b_s$.
		Conditioning on $T(X)$ therefore gives $T(X)\leq b_s$ on the former event and $T(X)\geq b_s$ on the latter.
		Since $\{X\in E_s\}=\{T(X)\leq s\}$ up to null sets, it follows that
		\begin{align*}
			\int_{\R}(b_s-x)_+\,d\nu(x)
			-\int_{\R}(b_s-m)_+\,d\mu(m)
			&=\E\!\left[(T(X)-X)\1_{\{T(X)\leq s\}}\right]\\
			&=\E\!\left[
			\bigl(T(X)-\E[X|T(X)]\bigr)
			\1_{\{T(X)\leq s\}}
			\right]=0.
		\end{align*}
		This contradicts \eqref{Eq: Local_strict_gap} at $z=b_s$. 
		Hence, $\ell_j>c_j$ must hold. This proves part 1.

		\subsection{Proof of Proposition~\ref{Prop: Contour_density_cdf}}\label{App. Proof_Contour_density_cdf}
		We only prove part 1; the proof of part 2 is similar and omitted.
		Throughout, $I_j(s)=[a_j(s),b_j(s)]$ is the interval matching the mass and first moment of the lower target contour $\alpha_s=\mu|_{(c_j,s]}$, as defined by \eqref{Eq: Lower_contour_mass} and \eqref{Eq: Lower_contour_barycenter}.
		We show that $a_j(s)$ is non-increasing and $b_j(s)$ is non-decreasing.
		
		First, any interval $[a,b]$ matching the mass and first moment of a submeasure $0\leq\alpha\leq\mu$ with $0<\alpha(\R)<1$ must satisfy
		\begin{equation}\label{Eq: Contour_candidate_interior}
			c_j<a<b<d_j.
		\end{equation}
		Indeed, the prior is atomless, so positive mass requires $a<b$.
		If $a=c_j$, matching mass and first moment gives
		\begin{equation*}
			\int_{\R}(b-x)_+\,d\nu(x)
			=\int_{\R}(b-m)\,d\alpha(m)
			\leq\int_{\R}(b-m)_+\,d\mu(m),
		\end{equation*}
		contradicting \eqref{Eq: Local_strict_gap} at $z=b\in(c_j,d_j)$.
		If $b=d_j$, the same moment equalities give
		\begin{equation*}
			\int_{\R}(x-a)_+\,d\nu(x)
			=\int_{\R}(m-a)\,d\alpha(m)
			\leq\int_{\R}(m-a)_+\,d\mu(m).
		\end{equation*}
		Since $\mu$ and $\nu$ have the same mean, the difference between these two positive-part integrals equals the difference in \eqref{Eq: Local_strict_gap} at $z=a$.
		This is again a contradiction, proving \eqref{Eq: Contour_candidate_interior}.
		
		Now fix $s\in(\ell_j,r_j)$ that is not a target atom.
		We first show
		\begin{equation}\label{Eq: Contour_lower_straddle}
			a_j(s)\leq s\leq b_j(s).
		\end{equation}
		The target restricted to $(c_j,s]$ has conditional mean at most $s$.
		The matching prior interval must have the same conditional mean, so $a_j(s)\leq s$.
		If $s>\tau_j$, the mass equation, $F(c_j)=G(c_j)$, and the cdf condition give
		\begin{equation*}
			F(b_j(s))=F(a_j(s))+G(s)-G(c_j)
			\geq G(s)\geq F(s),
		\end{equation*}
		and hence $b_j(s)\geq s$.
		
		If $s\leq\tau_j$, suppose instead that $b_j(s)<s$.
		Below $\ell_j$, only the prior restriction can contribute to the moment difference below, and its contribution is non-positive.
		On $(\ell_j,s)$, the density condition $f\leq g$ and the non-negative target atom masses also make the prior contribution no greater than the target contribution.
		Moreover, connectedness of the target support implies that $G$ assigns positive mass to $(\max\{b_j(s),\ell_j\},s)$, where the prior restriction has no mass.
		Therefore,
		\begin{equation*}
			\int_{a_j(s)}^{b_j(s)}(x-\ell_j)\,dF(x)
			-\int_{(c_j,s]}(m-\ell_j)\,dG(m)<0.
		\end{equation*}
		But the matching mass and first moment require this difference to be zero.
		This contradiction proves \eqref{Eq: Contour_lower_straddle}.
		
		By \eqref{Eq: Contour_candidate_interior}, both boundaries are interior.
		The Jacobian determinant of the two moment equations with respect to their boundaries is
		\begin{equation*}
			[a_j(s)-b_j(s)]f(a_j(s))f(b_j(s))\ne0.
		\end{equation*}
		The implicit function theorem therefore gives continuously differentiable boundaries between target atoms.
		Differentiating \eqref{Eq: Lower_contour_mass} and \eqref{Eq: Lower_contour_barycenter} gives
		\begin{equation}\label{Eq: Contour_lower_derivatives}
			a_j'(s)=\frac{g(s)[s-b_j(s)]}{[b_j(s)-a_j(s)]f(a_j(s))}\leq0
			\quad\text{and}\quad
			b_j'(s)=\frac{g(s)[s-a_j(s)]}{[b_j(s)-a_j(s)]f(b_j(s))}\geq0.
		\end{equation}
		Thus $I_j(s)$ expands between target atoms.
		
		It remains to check expansion when a target atom is included.
		Fix an atom $x_k$ of mass $q_k$ and add it gradually.
		For $t\in[0,q_k]$, let $[a(t),b(t)]$ match the mass and first moment of
		\begin{equation*}
			\alpha_t:=\mu|_{(c_j,x_k)}+\frac{t}{p_j}\delta_{x_k}.
		\end{equation*}
		This is a submeasure of $\mu$, so its matching interval exists by Lemma~\ref{Lemma: Moment_matched_interval}.
		In terms of the original cdfs, its boundaries satisfy
		\begin{align*}
			F(b(t))-F(a(t))&=G(x_k^-)-G(c_j)+t,\\
			\int_{a(t)}^{b(t)}x\,dF(x)&=\int_{(c_j,x_k)}m\,dG(m)+tx_k.
		\end{align*}
		Consider values of $t$ for which $0<\alpha_t(\R)<1$.
		The matching conditional mean is at most $x_k$, so $a(t)\leq x_k$.
		If $\ell_j<x_k\leq\tau_j$, the preceding moment comparison with $s=x_k$ and only mass $t$ of the atom included gives $b(t)\geq x_k$.
		This also covers an atom at $\tau_j$.
		If $x_k>\tau_j$, the cdf condition gives $F(x_k)\leq G(x_k^-)$; at $x_k=r_j$, this follows by taking left limits.
		Consequently,
		\begin{equation*}
			F(b(t))=F(a(t))+G(x_k^-)-G(c_j)+t
			\geq G(x_k^-)+t\geq F(x_k),
		\end{equation*}
		so $b(t)\geq x_k$.
		Finally, if $x_k=\ell_j$, the partial target consists only of mass at $x_k$, and matching its mean directly gives $a(t)\leq x_k\leq b(t)$.
		
		Thus $a(t)\leq x_k\leq b(t)$ in all cases.
		By \eqref{Eq: Contour_candidate_interior}, both boundaries are interior, and differentiating the two moment equations gives
		\begin{equation*}
			\frac{da(t)}{dt}=\frac{x_k-b(t)}{[b(t)-a(t)]f(a(t))}\leq0
			\quad\text{and}\quad
			\frac{db(t)}{dt}=\frac{x_k-a(t)}{[b(t)-a(t)]f(b(t))}\geq0.
		\end{equation*}
		Hence including the atom also expands the interval.
		At $t=0$, the matching restriction is the limit of the lower contour restrictions as $s\uparrow x_k$; at $t=q_k$, it is $\nu|_{I_j(x_k)}$.
		These endpoint identifications follow from boundedness of the component, continuity of the two moment integrals, and uniqueness in Lemma~\ref{Lemma: Moment_matched_interval}.
		Zero-mass and full-mass contours are represented by the empty set and the entire component.
		
		Since there are finitely many target atoms, the expansion between atoms and at each atom proves that $I_j(s)$ expands throughout $(c_j,d_j)$.
		Proposition~\ref{Prop: Nested_Contours} then gives the lower contour nested-interval implementation $T_L$ and its uniqueness up to $\nu$-null sets.
		This proves part 1.

		\subsection{Proof of Proposition~\ref{Prop: Contour_uniform_target}}\label{App. Proof_Contour_uniform_target}
		Let $\mu$ and $\nu$ denote the probability measures of $G$ and $F$, respectively.
		Concavity and non-uniformity imply $F(x)>x$ for every $x\in(0,1)$.
		Indeed, equality at an interior point of the chord joining $(0,0)$ and $(1,1)$ would force the concave cdf to coincide with this chord everywhere.
		Therefore, $\pi<1/2$. 
		On $[0,2\pi]$, define $h(s):=F(s)-s/{2\pi}$.
		This function is concave, with $h(0)=0$ and $h(2\pi)<0$.
		Since $F$ and $G$ have the same mean, we have $\int_0^1\left[F(s)-G(s)\right]ds=0$. 
		Moreover, $F(s)-G(s)<0$ on $(2\pi,1)$.
		Thus, $h$ must be positive somewhere on $(0,2\pi)$.
		Concavity then implies that there is a unique $\tau\in(0,2\pi)$ such that $F(s)>G(s)$ for $0<s<\tau$ and $F(s)<G(s)$ for $\tau<s<1$.
		Consequently,
		\begin{equation*}
			I_F(z)-I_G(z)=\int_0^z\left[F(s)-G(s)\right]ds>0,\qquad 0<z<1.
		\end{equation*}
		Together with equality of means, this proves $G\in\MPC(F)$ and identifies $(0,1)$ as its unique non-separating interval.
		
		We now verify part 2 of Proposition~\ref{Prop: Contour_density_cdf}.
		The cdf condition $F(s)\geq G(s)$ holds for $0<s<\tau$ by construction.
		For $\tau<s<2\pi$, concavity gives
		\begin{equation*}
			f(s)-\frac{1}{2\pi}=h'(s)
			\leq\frac{h(s)-h(\tau)}{s-\tau}<0.
		\end{equation*}
		Hence $f(s)\leq g(s)=1/{2\pi}$ for $\tau<s<{2\pi}$, as required.
		Proposition~\ref{Prop: Contour_density_cdf} therefore gives an upper contour nested-interval implementation, which is unique up to prior-null sets by Proposition~\ref{Prop: Nested_Contours_upper}.
		Integrating the upper target contour gives \eqref{Eq: SG_contour_mass} and \eqref{Eq: SG_contour_first_moment}.
		
		\Needspace{5\baselineskip}
		It remains to verify the strict inequalities and endpoint properties of the boundaries.
		Fix $s\in(0,2\pi)$ and write $[a,b]=[a(s),b(s)]$.
		The matching interval has prior mass strictly between zero and one. The argument for \eqref{Eq: Contour_candidate_interior}, applied to this component, therefore gives $0<a<b<1$.
		
		Next suppose $a\geq s$.
		The mass constraint implies $F(a)<s/{2\pi}$, since $F(b)<1$.
		Since concavity makes $f$ non-increasing,
		\begin{equation*}
			f(x)\leq f(a)\leq\frac{F(a)}a<\frac{1}{2\pi},\qquad x\geq a.
		\end{equation*}
		Let $\beta:=\nu|_{[a,b]}$ and $\alpha:=\mu|_{[s,2\pi]}$.
		Both measures are supported on $[s,1]$.
		On $[s,2\pi]$, $\beta$ has density strictly below $1/{2\pi}$, whereas $\alpha$ has density $1/{2\pi}$; above $2\pi$, only $\beta$ can have mass.
		It follows that
		\begin{equation*}
			\int_{[0,1]}(x-2\pi)\,d(\beta-\alpha)(x)>0,
		\end{equation*}
		contradicting the matching mass and first moment.
		Thus $a(s)<s$.
		
		For a non-increasing density, the conditional mean on an interval is at most its midpoint.
		Indeed, with $z=(a+b)/2$,
		\begin{equation*}
			\int_a^b(x-z)f(x)\,dx
			=\int_0^{(b-a)/2}t\left[f(z+t)-f(z-t)\right]dt\leq0.
		\end{equation*}
		The target contour has conditional mean $(s+2\pi)/2$, so matching gives $b(s)\geq s+2\pi-a(s)>2\pi$.
		This proves $0<a(s)<s<2\pi<b(s)<1$.
		
		Since both endpoints are interior and $f$ is continuous and positive there, the implicit function theorem applied to the two moment constraints gives continuously differentiable boundaries.
		Differentiating these constraints yields
		\begin{equation*}
			a'(s)=\frac{b(s)-s}{2\pi[b(s)-a(s)]f(a(s))}>0 
			\qquad \text{and} \qquad 
			b'(s)=-\frac{s-a(s)}{2\pi[b(s)-a(s)]f(b(s))}<0\ . 
		\end{equation*}
		As $s\downarrow0$, the interval mass tends to one, forcing $(a(s),b(s))\to(0,1)$.
		As $s\uparrow2\pi$, the interval mass tends to zero.
		Since every such interval contains $2\pi$ and $f$ is positive near $2\pi$, both endpoints converge to $2\pi$.
		Thus the two boundary functions are continuous and strictly monotone on their respective ranges, and their inverses give \eqref{Eq: SG_implementation}. 
	\end{appendices}

	\bibliographystyle{apsr}
	\bibliography{References}

@article{kamenica2011bayesian,
  title={Bayesian persuasion},
  author={Kamenica, Emir and Gentzkow, Matthew},
  journal={American Economic Review},
  volume={101},
  number={6},
  pages={2590--2615},
  year={2011},
  publisher={American Economic Association}
}

@Article{Dworczak2019,
  author    = {Dworczak, Piotr and Martini, Giorgio},
  title     = {The simple economics of optimal persuasion},
  journal   = {Journal of Political Economy},
  year      = {2019},
  volume    = {127},
  number    = {5},
  pages     = {1993--2048},
  publisher = {The University of Chicago Press Chicago, IL},
}

@Article{Kolotilin2022censorship,
  author    = {Kolotilin, Anton and Mylovanov, Timofiy and Zapechelnyuk, Andriy},
  title     = {Censorship as optimal persuasion},
  number    = {2},
  pages     = {561--585},
  volume    = {17},
  journal   = {Theoretical Economics},
  publisher = {Wiley Online Library},
  year      = {2022},
}

@Article{gentzkow2016rothschild,
  author  = {Gentzkow, Matthew and Kamenica, Emir},
  title   = {A Rothschild-Stiglitz approach to Bayesian persuasion},
  journal = {American Economic Review},
  year    = {2016},
  volume  = {106},
  number  = {5},
  pages   = {597--601},
}

@Article{Kleiner2021extreme,
  author    = {Kleiner, Andreas and Moldovanu, Benny and Strack, Philipp},
  title     = {Extreme points and majorization: Economic applications},
  journal   = {Econometrica},
  year      = {2021},
  volume    = {89},
  number    = {4},
  pages     = {1557--1593},
  publisher = {Wiley Online Library},
  doi       = {10.3982/ECTA18312},
}

@Article{Guo2019interval,
  author    = {Guo, Yingni and Shmaya, Eran},
  title     = {The interval structure of optimal disclosure},
  number    = {2},
  pages     = {653--675},
  volume    = {87},
  journal   = {Econometrica},
  publisher = {Wiley Online Library},
  year      = {2019},
}

@Article{Strassen1965existence,
  author    = {Strassen, Volker},
  title     = {The existence of probability measures with given marginals},
  number    = {2},
  pages     = {423--439},
  volume    = {36},
  journal   = {The Annals of Mathematical Statistics},
  publisher = {Institute of Mathematical Statistics},
  year      = {1965},
}

@Article{Arieli2023optimal,
  author    = {Arieli, Itai and Babichenko, Yakov and Smorodinsky, Rann and Yamashita, Takuro},
  title     = {Optimal persuasion via bi-pooling},
  number    = {1},
  pages     = {15--36},
  volume    = {18},
  journal   = {Theoretical Economics},
  publisher = {Wiley Online Library},
  doi       = {10.3982/TE4663},
  year      = {2023},
}

@Article{Nutz2024martingale,
  author    = {Nutz, Marcel and Wang, Ruodu and Zhang, Zhenyuan},
  title     = {Martingale transports and {Monge} maps},
  number    = {6},
  pages     = {5556--5577},
  volume    = {34},
  journal   = {The Annals of Applied Probability},
  publisher = {Institute of Mathematical Statistics},
  doi       = {10.1214/24-AAP2099},
  year      = {2024},
}

@Article{Roesler2017buyer,
  author    = {Roesler, Anne-Katrin and Szentes, Bal{\'a}zs},
  title     = {Buyer-optimal learning and monopoly pricing},
  number    = {7},
  pages     = {2072--2080},
  volume    = {107},
  journal   = {American Economic Review},
  publisher = {American Economic Association 2014 Broadway, Suite 305, Nashville, TN 37203},
  doi       = {10.1257/aer.20160145},
  year      = {2017},
}

@Article{Zeng2023derandomization,
  author       = {Zeng, Yishu},
  journal      = {Journal of Economic Theory},
  title        = {Derandomization of persuasion mechanisms},
  pages        = {105690},
  volume       = {212},
  publisher    = {Elsevier},
  doi          = {10.1016/j.jet.2023.105690},
  year         = {2023},
}

@Unpublished{Kolotilin2024monotone,
  author = {Kolotilin, Anton and Li, Hongyi and Zapechelnyuk, Andriy},
  title  = {On Monotone Persuasion},
  note   = {Working paper, available at \url{https://arxiv.org/abs/2412.14400}.},
  year   = {2024},
}

@Unpublished{Sun2026competition,
  author = {Sun, Junze},
  note   = {Working paper, available at \url{https://papers.ssrn.com/sol3/papers.cfm?abstract\_id=5587750}.},
  title  = {On competition in linear persuasion games},
  year   = {2026},
}

@Article{SapiroGheiler2024,
  author  = {Sapiro-Gheiler, Eitan},
  title   = {Persuasion with ambiguous receiver preferences},
  journal = {Economic Theory},
  year    = {2024},
  volume  = {77},
  number  = {4},
  pages   = {1173--1218},
  doi     = {10.1007/s00199-023-01522-z},
}

@Article{BeiglbockJuillet2016,
  author  = {Beiglb{\"o}ck, Mathias and Juillet, Nicolas},
  title   = {On a problem of optimal transport under marginal martingale constraints},
  journal = {The Annals of Probability},
  year    = {2016},
  volume  = {44},
  number  = {1},
  pages   = {42--106},
  doi     = {10.1214/14-AOP966},
}

@Article{Kolotilin2025persuasion,
  author  = {Kolotilin, Anton and Corrao, Roberto and Wolitzky, Alexander},
  title   = {Persuasion and matching: Optimal productive transport},
  journal = {Journal of Political Economy},
  year    = {2025},
  volume  = {133},
  number  = {4},
  pages   = {1334--1381},
  doi     = {10.1086/734095},
}

@Article{Chen2025information,
  author    = {Chen, Yi-Chun and Hu, Gaoji and Yang, Xiangqian},
  journal   = {International Economic Review},
  title     = {Information design in allocation with costly verification},
  year      = {2025},
  number    = {3},
  pages     = {1267--1285},
  volume    = {66},
  publisher = {Wiley Online Library},
}

@Unpublished{Chen2025explainable,
  author = {Chen, Yiling and Lin, Tao and Tang, Wei and Tucker-Foltz, Jamie},
  note   = {Available at arXiv preprint arXiv:2508.14196},
  title  = {Explainable information design},
  year   = {2025},
}

@Unpublished{Arieli2025constructive,
  author = {Arieli, Itai and Babichenko, Yakov and Sandomirskiy, Fedor},
  note   = {Working paper.},
  title  = {{Constructive Blackwell's Theorem}},
  year   = {2025},
}

@Article{Chen2019equivalence,
  author    = {Chen, Yi-Chun and He, Wei and Li, Jiangtao and Sun, Yeneng},
  journal   = {Econometrica},
  title     = {Equivalence of stochastic and deterministic mechanisms},
  year      = {2019},
  number    = {4},
  pages     = {1367--1390},
  volume    = {87},
  publisher = {Wiley Online Library},
}
	
\end{document}